\documentclass[a4paper,UKenglish,cleveref, autoref, thm-restate, final]{lipics-v2021}
\usepackage{color,xspace}
\usepackage[ruled,linesnumbered]{algorithm2e}
\usepackage{float}
\usepackage{hyperref}

\usepackage[T1]{fontenc}
\usepackage{graphicx}
\usepackage{cite}
\usepackage{mathtools}

\usepackage[appendix=append]{apxproof}

\newcommand{\ignore}[1]{}

\newcommand{\eg}{e.g.~}
\newcommand{\MPMGAD}{MPMD-Set\xspace}
\newcommand{\SSCA}{Sensible-ALG\xspace}

\newcommand{\MPMGADMTS}{MPMD-Set-MTS\xspace}\newcommand{\MPMDSize}{MPMD-Size\xspace}
\newcommand{\ProcessCost}{processing cost\xspace}
\newcommand{\MPMDSizeMTS}{MPMD-Size-MTS\xspace}
\DeclareMathOperator{\timedist}{\mathit{d_{time}}} \DeclareMathOperator{\timedistc}{\mathit{d_{ctime}}} \DeclareMathOperator{\dist}{\mathit{d}} \DeclareMathOperator{\optcost}{optcost}
\DeclareMathOperator{\Sur}{sur}
\DeclareMathOperator{\pos}{pos}
\DeclareMathOperator{\atime}{atime}
\DeclareMathOperator{\cost}{cost}
\DeclareMathOperator{\reqGrowth}{reqGrowth}

\DeclareMathOperator{\symdif}{\triangle}
\DeclareMathOperator{\diff}{diff}

\hideLIPIcs \title{Online Matching with Set and Concave Delays} 

\author{Lindsey Deryckere}{School of Computer Science, The University of Sydney, Australia}{lindsey.deryckere@sydney.edu.au}{}{}

\author{Seeun William Umboh}{School of Computing and Information Systems, The University of Melbourne, Australia \and \url{http://williamumboh.com}}{william.umboh@unimelb.edu.au}{https://orcid.org/0000-0001-6984-4007}{}

\authorrunning{L. Deryckere and S.\.W. Umboh}

\Copyright{Lindsey Deryckere and Seeun William Umboh}



\nolinenumbers 

\EventEditors{Nicole Megow and Adam D. Smith}
\EventNoEds{2}
\EventLongTitle{Approximation, Randomization, and Combinatorial Optimization. Algorithms and Techniques (APPROX/RANDOM 2023)}
\EventShortTitle{\mbox{\scriptsize APPROX/RANDOM 2023}}
\EventAcronym{APPROX/RANDOM}
\EventYear{2023}
\EventDate{September 11--13, 2023}
\EventLocation{Atlanta, Georgia, USA}
\EventLogo{}
\SeriesVolume{275}
\ArticleNo{17} 
\begin{document}
\maketitle
 \begin{abstract}
	We initiate the study of online problems with \emph{set delay}, where the delay cost at any given time is an arbitrary function of the set of pending requests. In particular, we study the online min-cost perfect matching with set delay (\MPMGAD) problem, which generalises the online min-cost perfect matching with delay (MPMD) problem introduced by Emek et al. (STOC 2016). In MPMD, $m$ requests arrive over time in a metric space of $n$ points. When a request arrives the algorithm must choose to either match or delay the request. The goal is to create a perfect matching of all requests while minimising the sum of distances between matched requests, and the total delay costs incurred by each of the requests. In contrast to previous work we study \MPMGAD in the \emph{non-clairvoyant} setting, where the algorithm does not know the future delay costs. We first show no algorithm is competitive in $n$ or $m$. We then study the natural special case of \emph{size-based} delay where the delay is a non-decreasing function of the number of unmatched requests. Our main result is the first non-clairvoyant algorithms for online min-cost perfect matching with size-based delay that are competitive in terms of $m$. In fact, these are the first non-clairvoyant algorithms for any variant of MPMD. A key technical ingredient is an analog of the symmetric difference of matchings that may be useful for other special classes of set delay. Furthermore, we prove a lower bound of $\Omega(n)$ for any deterministic algorithm and $\Omega(\log n)$ for any randomised algorithm. These lower bounds also hold for clairvoyant algorithms. Finally, we also give an $m$-competititve deterministic algorithm for uniform concave delays in the clairvoyant setting.
\end{abstract} \section{Introduction}

Studying online problems with delay is a line of work that has recently gained traction in online algorithms (\eg \cite{emekoriginal, bin-packingClusteringDelays, cachingDelay, setCoverDelayNon-Clairvoyance}). In such problems, request arrive over time requiring service. Delaying the service of a request accumulates a delay cost given by a delay function associated with the request. The total cost of a solution is the cost of servicing all requests plus the sum of all delay costs incurred by each request.

We initiate the study of online problems with \emph{set delay}. In this model, we generalize the notion of delay to one where the instantaneous delay cost at any point in time is determined by an arbitrary monotone non-decreasing function of the set of pending requests, rather than the sum of individual delay functions associated with each request. In particular, we study the online min-cost perfect matching with set delay (\MPMGAD) problem, which generalizes of the min-cost perfect matching with delays (MPMD) problem introduced by Emek et al.~\cite{emekoriginal}. 

In MPMD, $m$ requests arrive over time in a metric space of $n$ points. Upon arrival of a request the algorithm must choose to either match the request, incurring a cost equal to the distance between the two requests, or to delay the request, incurring a cost given by a delay function associated with the request. 
Prior results for MPMD have mostly focused on each request sharing the same delay function (in particular, linear, concave, and convex) and achieve competitive ratios that solely depend either on $n$ or $m$. Moreover, existing algorithms rely on clairvoyance, where the algorithm has full knowledge of future delay costs. Furthermore, existing randomised algorithms rely on metric embeddings which require knowledge of the metric space in advance.

In this paper, our main contribution is to study the more general \MPMGAD in the least restrictive setting where the algorithm does not know the metric space in advance and has no knowledge of future delay costs. 
We begin by showing that, in contrast to prior results, the \MPMGAD problem does not admit a deterministic competitive ratio that solely depends on $n$ or $m$.
\begin{theorem}
	\label{thm:lb_general}
	Every deterministic algorithm for \MPMGAD has competitive ratio $\Omega(\Phi)$, where $\Phi$ is the aspect ratio of the metric space.
\end{theorem}

Our lower bound holds even for simple instances where $n$ and $m$ are constants.
Thus, we restrict our attention to designing a competitive solution for the \MPMGAD problem where the instantaneous delay cost at any point in time is a monotone non-decreasing function of the number of unmatched requests at that time. We call such a delay cost function \emph{size-based} (See Section~\ref{sec:preliminaries} for a formal definition). \MPMGAD with size-based delay (\MPMDSize) has natural applications in practical settings with service-level agreements such as cloud computing.\footnote{In these settings, the service level agreement requires the cloud provider to provide a certain level of service and the provider incurs penalties if the level is not met.}

Our main result is the first competitive algorithms for \MPMDSize, where the competitive ratio is a function of the number of requests. At the core of our result is a reduction from \MPMDSize to the well-known Metrical Task System (MTS) problem (defined in Section~\ref{ssec:techniques}).
\begin{theorem}
	\label{thm:size-based-to-match}
	For any $f(N)$-competitive algorithm for MTS with $N$ states, there is an $f(2^m)$-competitive algorithm for \MPMDSize.
\end{theorem}
We obtain our main result by applying state-of-the-art algorithms for MTS with some modifications. \begin{corollary}
	\label{cor:size-based-to-match}
	For \MPMDSize, there is an $O(2^m)$-competitive deterministic algorithm and an $O(m^4)$-competitive randomised algorithm. \end{corollary}
We emphasise that our algorithms are non-clairvoyant and do not need to know the metric space in advance. To the best of our knowledge, this is the first non-clairvoyant online algorithm for this problem.
Non-clairvoyant algorithms nevertheless have been designed for other online problems such as the Set Cover problem~\cite{setCoverDelayNon-Clairvoyance}, the $k$-server problem~\cite{nonclairvoyantKServer}, and multi-level aggregation~\cite{soda2023}. We also remark that every deterministic algorithm for known variants of online matching with delays has a competitive ratio that depends on $m$.

We complement Corollary~\ref{cor:size-based-to-match} with the following lower bounds. 
\begin{theorem}
	\label{thm:det_lb_size-based}
	 Every deterministic algorithm for \MPMDSize has competitive ratio $\Omega(n)$.\end{theorem}
\begin{theorem}
	\label{thm:rand_lb_size-based}
	Every randomised algorithm for \MPMDSize has competitive ratio $\Omega(\log n)$.\end{theorem}

The proofs of theorems~\ref{thm:det_lb_size-based} and~\ref{thm:rand_lb_size-based} can be found in the Appendix.

Finally, we consider MPMD with uniform concave delay in the clairvoyant setting and give the first deterministic algorithm for it. In this problem, we are given a non-negative, non-decreasing concave function $f$. The delay cost incurred by a request $r$ is $f(w_r)$ where $w_r$ is the time between $r$'s arrival and when it was matched. The total delay cost is the sum of the delay cost of each request. 

\begin{theorem}
    \label{thm:concave}
    There exists an $O(m)$-competitive deterministic algorithm for MPMD with concave delay.
\end{theorem}
The correctness and competitiveness of our algorithm only relies on the fact that the time-augmented space satisfies the properties of a metric space. Similar to previous deterministic solutions for uniform linear delay, our algorithm does not need the metric space to be finite, and does not need to know it in advance. 

\subsection{Our Techniques}
\label{ssec:techniques}

Our main technical contribution is an online reduction from the \MPMGAD problem to MTS, which constitutes the proof of Theorem~\ref{thm:size-based-to-match}. The Metrical Task System (MTS) problem, introduced by Borodin et al. \cite{originalMTS}, is a cost minimisation problem defined by a set of states $S = \{s_1, s_2,..., s_k\}$ and a cost matrix $c$ that defines the cost of moving between states. The input consists of an initial state $S_0$ and a sequence of tasks $T = (t_1,..., t_\ell) $. Each task $t_j$ is associated with a $k$-dimensional cost vector $C_{j}$ whose $i$-th coordinate defines the cost of servicing task $t_j$ in state $s_i$. For a given input task sequence $T$, a solution is a sequence of states (called a \emph{schedule}) $\sigma = (S_1, S_2, ..., S_\ell)$, where $S_j$ is the state that task $j$ is processed in. The total cost of a schedule consists of the costs associated with moving states (\emph{transition cost}), as well as the cost of processing the tasks (\emph{processing cost}).The aim is to produce a schedule of minimum cost.

We briefly outline the three main parts of the reduction below. 

\subparagraph*{Step 1: \MPMGAD to MTS}
The first part of the reduction transforms an instance of \MPMGAD into an instance of MTS. A natural approach at a reduction to MTS is to use the set of all possible matchings of requests as the set of states for the MTS instance. The transition cost between two states is then the total length of the edges in the symmetric difference of the corresponding matchings. Finally, there is a task for each timestep in the matching problem and the cost of processing the task in a state is the instantaneous delay incurred by the set of unmatched requests. Unfortunately, the number of states is equal to the number of possible matchings between the requests which is $\sim (\frac{m}{e})^{m/2} \frac{e^{\sqrt{m}}}{(4e)^{1/4}}$ for $m$ requests.

Instead, we use the set of all possible even-sized subsets of the requests as the set of MTS states. Each state represents a set of requests that are matched. The set of input states thus develops over time as more requests arrive. The initial state is the empty set. 

We now define the transition costs of the MTS. We define a \emph{transition graph} $G(V,E)$ where $V$ is the set of even-sized subsets of requests. The transition graph $G$ has an edge between two states $S$ and $S'$ if $S \subset S'$ and $|S'| = |S| + 2$. In other words, $S_2$ consists of the same requests as $S_1$ with 2 additional requests, say $p,q$. The cost of the edge between the two states is $d(p,q)$, the distance between $p$ and $q$ in the original \MPMGAD instance. The transition cost between any two states in the MTS instance is defined to be the minimum cost path between the two corresponding nodes in $G$. The delay cost is translated into the vector costs associated with serving tasks: for any timestep $t$, the cost of servicing a task in state is simply the instantaneous delay cost accumulated by the set of requests that have arrived so far in the original \MPMGAD instance, that are not in that state.

Henceforth, we refer to an instance of MTS that is reduced from \MPMGAD as \emph{\MPMGADMTS}, and an MTS instance that is reduced from \MPMGAD with size-based delay as \emph{\MPMDSizeMTS}.

\begin{definition}[Monotone]
	A schedule $\sigma = (S_1, \ldots, S_\ell)$ is \emph{monotone} if every transition only adds requests to the current state, i.e.~$S_{i-1} \subseteq S_i$ for every $i$. In other words, the path never involves moving to a strictly smaller state. An algorithm for \MPMGADMTS is \emph{monotone} if it always produces a \emph{monotone} schedule.
\end{definition}
It is easy to see that a monotone schedule can be converted into a solution for the \MPMGAD instance without any increase in cost: when the schedule adds a set of requests $S'$ to its current state, we add a min-cost perfect matching on $S'$ to our matching. However, when we run existing MTS algorithms on the \MPMGADMTS instance, there is no guarantee that they will return a monotone schedule. 

\subparagraph*{Step 2: Converting to Monotone}
Fortunately, it is possible to convert, in an online manner, an arbitrary \MPMDSizeMTS solution to a monotone solution at no extra cost.
We do this by designing an online algorithm which, given an online sequence of states $\sigma = (S_1, \ldots, S_{\ell})$, produces for each state $S_i$ a corresponding state $S^\prime_i$ such that the resulting schedule produced by the algorithm is \emph{monotone}. We refer to an algorithm that \emph{transforms} a given state as a \emph{state conversion algorithm}. 

Since the instantaneous delay in the \MPMDSize instance is a non-increasing function of currently matched requests, for every task in the \MPMDSizeMTS instance created by the reduction, the processing cost is a monotone non-increasing function of the state size. Our algorithm will exploit this by maintaining the invariant that our state is always at least large as that of $\sigma$. The technical crux here is to show that when our state is smaller, we can augment our state in a cost-efficient manner. If the MTS states were matchings, we can consider augmenting paths in the symmetric difference of our matching and that of $\sigma$. Motivated by this, we define analogs of the symmetric difference and augmenting paths to augment our state. We believe these ideas are useful for other interesting special classes of set delay functions.

\subparagraph*{Step 3: Applying MTS algorithms} There are two issues that prevent us from applying existing MTS algorithms directly. First, the cost bounds of all known algorithms for MTS have an additive term that is equal to the diameter of the MTS state space, and the MTS instance created by our reduction has state space with diameter much larger than the optimal. The second issue is that our reduction creates an MTS instance whose state space is constructed online, i.e.~the states arrive over time. At a high level, the first issue can be overcome by a guess-and-double approach. The second issue is a problem for randomised MTS algorithms that rely on embedding the state space into a tree as a pre-processing step. This is overcome by using the online embedding of \cite{onlineEmbedding}. See Section~\ref{ssec:analysis} for a more detailed discussion.

\subparagraph*{Designing a deterministic algorithm for MPMD with Concave Delay}
We use the (offline) moat-growing framework (generally used for constrained connectivity problems) by Goemans and Williamson~\cite{forest}, to design an online deterministic linear programming-based algorithm for MPMD with uniform concave delay. Our algorithm is a modification of an existing linear programming-based algorithm due to Bienkowski et al.~\cite{pd}, who give a deterministic competitive algorithm for MPMD with uniform linear delay. Their algorithm heavily relies on the fact that, when the delay function is linear, requests accumulate delay cost at the same rate at all times, regardless of their arrival time. The main challenge in applying the framework to concave delay is that, unlike in the case of linear delay, the requests can accumulate delay at different rates at any point in time, depending on their arrival time. See Section~\ref{sec:concave} for a detailed discussion.
\subsection{Related Work}
MPMD was introduced by Emek et al.~\cite{emekoriginal} where the delay functions associated with each request are uniform linear. They designed a randomised algorithm that achieves a competitive ratio of $O(\log^2 n + \log \Delta)$, where $n$ is the number of points in the metric space and $\Delta$ is its aspect ratio. Azar et al.~\cite{polylog} used a randomised HST embedding to provide a $O(\log n)$-competitive almost-deterministic algorithm, improving Emek et al.'s bound and removing the dependency on the aspect ratio of the metric space. Furthermore, they provided a lower bound of $\Omega(\sqrt{\log n})$ for any randomised algorithm in the case of linear delay. Ashlagi et al.~\cite{bipartite} improved this lower bound to $\Omega(\frac{\log n}{\log \log n})$ and $\Omega(\sqrt{\frac{\log n}{\log \log n}})$ for the bipartite case, which are the best known so far. Liu et al. furthermore adapted the algorithm by Azar et al. to the bipartite setting and improved the analysis of Emek et al.'s algorithm to $O(\log n)$. The next deterministic algorithm for simple metrics was by Emek et al.~\cite{2sources} who proved a competitive ratio of 3 for the simple metric space of 2 points. The first deterministic algorithm for general metric spaces was by Bienkowski et al.~\cite{spheres} and their analysis resulted in a competitive ratio of $O(m^{2.46})$. Bienkowski et al.~\cite{pd} and Azar et al.~\cite{hemispheres} concurrently and independently improved this bound to $O(m)$ and $O(m^{0.59})$ respectively, introducing the first linear and sub-linear deterministic solutions to the problem. The algorithms above assumed the delay cost to be given by a uniform linear delay function associated with each individual request. 

Liu et al.~\cite{convex} was the first to consider convex delay functions and demonstrated an interesting gap between the solutions for the case with linear delay and convex delay on a uniform metric space by giving a deterministic asymptotically optimal $O(m)$-competitive algorithm for the uniform metric space.

Azar et al.~\cite{concave} subsequently considered the problem with concave delay and achieved an $O(1)$-competitive deterministic algorithm for the single point metric space and an $O(\log n)$ randomised algorithm for general metric spaces. 

The above algorithms assumed all requests incurred delay in accordance with uniform delay functions and regarded the delay function to be associated with each individual request. Furthermore, all prior solutions to MPMD assumed clairvoyance. To the best of our knowledge, no one has considered the non-clairvoyant generalisation of the problem where the delay function depends on the set of unmatched requests.

Non-clairvoyant algorithms nevertheless have been designed for other online problems such as the Set Cover problem~\cite{setCoverDelayNon-Clairvoyance, soda2023}, the $k$-server problem~\cite{nonclairvoyantKServer}, and multi-level aggregation~\cite{soda2023}.

The notion of introducing delay to online problems originated well before it was applied to online metric matching and finds applications in amongst others aggregating messages in computer networks, aggregating orders in supply-chain management, and operating systems. See~\cite{WilliamsPaper, bin-packingClusteringDelays, bin-packingDelay, cachingDelay, multiLevelAggregation, multiLevelAggregationDepth, networkDesignDelay, networkDesignDelayLowerBound, onlineService, onlineServiceLine, setCoverDelayNon-Clairvoyance, soda2023} for further reading. All problems above define the cost of delay as a function associated with each request. To the best of our knowledge, no online problems with delay have so far defined the cost of delay as an arbitrary function of the set of unmatched requests. \section{Preliminaries}
\label{sec:preliminaries}
In this section we introduce our notation and give formal definitions for set delay and size-based delay functions, as well as \MPMGAD.

\subsection{Min-cost Perfect Matching with Delay}
The min-cost perfect matching with delays (MPMD) problem, introduced by Emek et. al~\cite{emekoriginal} is defined on a metric space $(V,d)$, which consists of a set of points $V$ and distance function $d: V\times V \rightarrow \mathbb{R}^+$.
An online input instance over $(V,d)$ is a sequence of requests $R = (r_1, ..., r_{m})$ that arrive at points in the metric space over time. Each $r_k \in R$ has an associated position and arrival time. We assume, without loss of generality, that time is divided into discrete timesteps.

Upon the arrival of a request, the algorithm must choose to either match the request, incurring a cost equal to the distance between the two requests in the metric space, or to delay the request, incurring a cost given by a delay function associated with the request, in the hope of finding a more suitable match in the near future. 

A solution produced by an online matching algorithm is a sequence of matchings $\mathbb{M} = (M_0...M_{final})$, where $M_i$ is the matching associated with the $i$th timestep. Note that we assume that requests only arrive at the start of a timestep. A solution $\mathbb{M}$ must satisfy the following properties:
\begin{itemize}
	\item $M_0 = \emptyset$
	\item $M_{final}$ is a perfect matching 
	\item For all $i$, $M_{i} \subseteq M_{i+1}$ 
\end{itemize}
We refer to the third property as \emph{monotonicity}.
The cost associated with a solution $\mathbb{M}$ consists of the sum of the distances between matched requests in $M_{final}$ plus the sum of the delay costs incurred by all requests. 
The aim of an online matching algorithm is to produce a sequence of matchings that satisfies the above properties with minimal cost. 

In the original MPMD problem, the delay cost incurred by a request is the time between its arrival and when it was matched. In \MPMGAD, the instantaneous delay incurred at a timestep $t$ can be an arbitrary function of the set of currently unmatched requests $U_t$. The total delay cost is the sum over timesteps $t$ of $f_t(U_t)$ where $f_t$ is a set delay function as defined below.
\begin{definition}[Set delay function]
    Let $U$ be a set of requests. We define a delay function $f_t:2^U \rightarrow \mathbb{R}^{\geq 0}$, for any timestep $t$, to be a \emph{set delay} function if it satisfies the following properties:
\begin{itemize}
	\item $f_t(\emptyset) = 0$
	\item $A \subseteq B \Rightarrow f_t(A) \leq f_t(B)$
	\item For all $\emptyset \neq U \in 2^V$, we have $\sum_{t=0}^\infty f_t(U) = \infty$	
\end{itemize}
The last property implies that all requests must eventually be matched.
\end{definition}

In \MPMDSize, the set delay function is a size-based delay function, as defined below.
\begin{definition}[Size-based delay function]
	We define a delay function $f_t:2^U \rightarrow \mathbb{R}^{\geq 0}$ to be \emph{size-based} if, for any timestep $t$, it satisfies all properties of a set delay function and is monotone non-decreasing as a function of the size of the set of requests $U$.
\end{definition}

\section{A Lower Bound for \MPMGAD}
In this section, we prove Theorem~\ref{thm:lb_general}.

\begin{proof}[Proof of Theorem~\ref{thm:lb_general}]
Consider a four-point metric space (as depicted in figure~\ref{fig:lb_MPMGAD}) with three points at distance $\epsilon$ from one another, and the fourth point ($p_4$) at distance $D$ from the other points, where $D$ is the diameter of the metric space.
\begin{figure}[H]
	\begin{center}
		\includegraphics[scale=0.3]{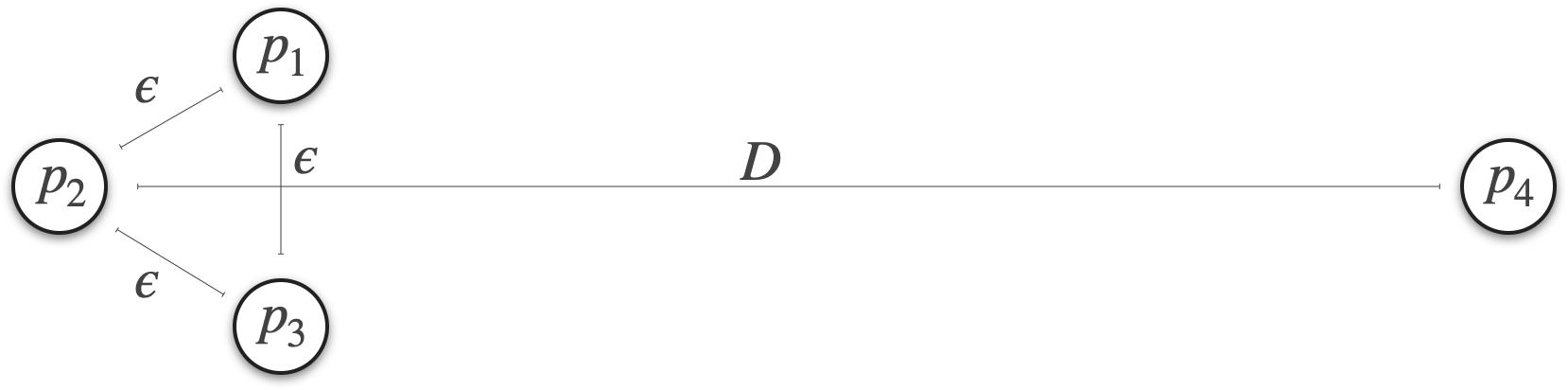}
		\caption{A visualisation of the four-point metric space}
		\label{fig:lb_MPMGAD}
	\end{center}
\end{figure}
We define a request sequence of six requests $R = (r_1, r_2, r_3, r_4, r_5, r_6)$ where the first four requests arrive at time $t = 0$ and the latter two arrive at time $t = 2$. For each $i \in \{1...4\},$ we place request $r_i$ on point $p_i$. At $t=2$, we then place request $r_5$ on $p_3$ and $r_6$ on $p_4$.
In terms of delay, we are working with the special case of deadline functions. A \emph{deadline function} is a delay function that is 0 up until some time $d$, called the deadline, and $\infty$ afterwards. When the deadline of a request is reached, the algorithm must ensure that the request is matched.

At $t=0$, request $r_1$ reaches its deadline and hence the algorithm will need to match two requests. Since the algorithm is non-clairvoyant, the algorithm has no knowledge of the deadlines of future requests. We therefore assume without loss of generality that it matches $r_1$ to $r_3$ and pays a distance cost of $\epsilon$. At $t=1$, $r_2$ reaches its deadline and the algorithm is forced to match it to $r_4$ at a distance cost of $D$. At $t=2$, the final two requests will arrive and instantly reach their deadline. The algorithm will consequently need to match $r_5$ to $r_6$ at a distance cost of $D$. The total cost of ALG is $2D + \epsilon$.

The optimal offline solution OPT is to match $r_1$ to $r_2$ and match locally at $t=2$ on points $p_3$ and $p_4$. The total cost of OPT is therefore $\epsilon$. The competitive ratio of the algorithm is $\Omega(D/\epsilon)$ and $D/\epsilon$ is the aspect ratio of the metric space, as desired.
\end{proof}
	\label{sec:reduction}
	In this section, we prove Theorem~\ref{thm:size-based-to-match} by defining a reduction from \MPMGAD to MTS. 
We start by translating an arbitrary instance of \MPMGAD into an instance of MTS in Section~\ref{ssec:reduction}. In Section~\ref{ssec:monotone}, we show that we can transform an arbitrary \MPMDSizeMTS solution into a monotone solution of the same or less cost.
As observed in the Introduction, a monotone schedule directly corresponds to a solution for the online matching with delay problem of equal cost. This completes the proof of Theorem~\ref{thm:size-based-to-match}. We finish this section with a proof of Corollary~\ref{cor:size-based-to-match}.

\subsection{Translating an instance of \MPMGAD into and instance of MTS}\label{ssec:reduction}
We define the set of internal states of the MTS instance to be the set of all possible subsets of the requests that have arrived so far in the \MPMGAD instance. In order to define the transition cost associated with moving between states, we define the following graph.
\begin{definition}[Transition graph $G$]
    The nodes of $G$ are the states of the \MPMGADMTS instance. We define an undirected edge $(S,S') \in E$ between states $S$ and $S'$ if $S = S' \cup \{r, r'\}$ and define its cost $c(S,S')$ to be $d(r, r')$. Paths in the transition graph are called \emph{transition paths}.
\end{definition}
We define the transition cost $c(S, S')$ of moving between arbitrary states $S$ and $S'$ in the MTS instance to be the cost of the shortest path between them in the transition graph $G$. Note that a path in $G$ between states $S, S'$ corresponds to a sequence of transitions from $S$ to $S'$, where each edge in the path corresponds to a transition that either adds two requests or removes two requests. Also, note that any schedule $\sigma$, produced by an online MTS algorithm, corresponds to a path in the metric completion of $G$ and the total transition cost incurred by $\sigma$ is simply the total cost of the path. Each task in the MTS instance is associated with a given timestep in the \MPMGAD problem. The cost vector associated with each task is, for every state $S$, the instantaneous delay cost accumulated by the set of requests not in $S$.
\begin{definition} [$R_i$]
        We define $R_i$ to be the set of requests that have arrived up to and including timestep $i$ in the original \MPMGAD instance.
\end{definition}
The total cost of a schedule $\sigma = (S_0,...,S_T)$ can be expressed as follows.
\begin{align*}
	\cost(\sigma) = \sum_{i=0}^{T-1} c(S_i, S_{i+1}) + f_i(R_i \setminus S_i).
\end{align*}

By construction, the cost associated with processing the tasks represent the delay cost incurred by the requests, while the distance cost is represented by the transition costs associated with moving between states. 
 \subsection{Size-based delay functions admit monotone scheduling algorithms}\label{ssec:monotone}
In this subsection, we prove the existence of an online algorithm that converts an arbitrary \MPMDSizeMTS solution into a monotone solution, without incurring any extra cost. 
\begin{lemma}
	\label{lem:any_to_monotone}
	There exists an online algorithm that converts an arbitrary \MPMDSizeMTS solution into a monotone solution of the same or less cost. 
\end{lemma}

\begin{proof}
	To prove Lemma~\ref{lem:any_to_monotone},  we define an online \emph{state conversion algorithm} (\SSCA) which, for every state $S_i$ in a schedule $\sigma$, produced by an \emph{arbitrary} online scheduling algorithm (OSA), produces a state $S^\prime_i$ such that the cost of the schedule $\sigma^\prime = (S^\prime_0, \ldots , S^\prime_{|\sigma|})$ is at most the cost of the original schedule $\sigma$, and $\sigma^\prime$ is monotone.

	The state conversion algorithm aims to maintain the invariant that $|S'_i| \geq |S_i|$ and the main property of a monotone schedule, which is that $S^\prime_{i-1} \subseteq S^\prime_i$ for every $i$. At a high level, it does this by adding requests to its current state when the invariant is violated that allows it to also move closer to the current state of $\sigma$. The algorithm \SSCA uses the following analog of the symmetric difference of two matchings to augment its current state in a cost-efficient manner.

 \begin{definition}
     [$(A,B)$-difference graph]
     Let $H$ be a (multi-)graph with vertex set $R$, and $A,B$ be states. The graph $H$ is an \emph{$(A,B)$-difference graph} if it is a $T$-join with $T = A \symdif B$: for every $r \in R$, the degree of $r$ in $H$ is odd if and only if $r \in A \symdif B$.
 \end{definition}
For example, one way to obtain an $(A,B)$-difference graph is by taking a perfect matching $M_A$ on $A$ and a perfect matching $M_B$ on $B$ and then taking the symmetric difference of $M_A$ and $M_B$.

\begin{definition}
    [$P$-difference graph and realisable difference graphs]
        Let $A,B$ be states and $P$ be a transition path between $A$ and $B$. The \emph{$P$-difference graph} $\diff(P)$ is a (multi-)graph on vertex set $R$ where the multiplicity of the edge $(p,q)$ is equal to the number of transitions along the path $P$ that adds or removes the $\{p,q\}$. An $(A,B)$-difference graph $H$ is \emph{realisable} if it is the $P$-difference graph for a transition path $P$ between $A,B$.
    \end{definition}
Note that if $P$ is a transition path between states $A,B$, then a $P$-difference graph is an $(A,B)$-difference graph. Moreover, if $P$ is a shortest path between $S$ and $S'$, then the total cost of the edges in $\diff(P)$ is exactly equal to $c(S,S')$. We also note that not all difference graphs are realisable\footnote{For example, consider the difference graph $H$ consisting of requests $p,q,r$ and edges $(p,q)$ and $(q,r)$. Observe that $H$ is an $(\{p,r\},\emptyset)$-difference graph but there is no transition path $P$ from $\{p,r\}$ to $\emptyset$ such that $\diff(P) = H$.}

We now characterise the structure of difference graphs that correspond to shortest paths in the transition graph. 

\begin{definition}
    [Canonical difference graphs]
    Let $H$ be an $(A,B)$-difference graph.  We say that $H$ is \emph{canonical} if it can be decomposed into a collection of $\ell := |A \symdif B|/2$ edge-disjoint paths $Q_1, \ldots, Q_\ell$ between disjoint pairs $(p_i,q_i)$ of $A \symdif B$ such that for each $i$, $Q_i$ consists of a single edge $(p_i, q_i)$ if either both $p,q$ are in $A \setminus B$ or both in $B \setminus A$, and $Q_i$ consists of two edges $(p,s), (s,q)$ for some other request $s$ if exactly one of $p,q$ is in $A \setminus B$ and the other in $B \setminus A$.
\end{definition}    

\begin{proposition}
    \label{prop:realisable}
    If $H$ is a canonical $(A,B)$-difference graph, then $H$ is realisable.
\end{proposition}

\begin{proof}
    We show by induction on $\ell = |A \symdif B|/2$ that there is a transition path $P$ from $A$ to $B$. Consider the base case $\ell = 1$. Suppose $Q_1$ consists of a single edge $(p,q)$. If $p,q \in A \setminus B$, then we can transition directly from $A$ to $B$ by removing $p,q$; otherwise, we add $p,q$. Next, suppose that $Q_1$ consists of two edges $(p,s), (s,q)$. Consider the case that $p \in A \setminus B$ and $q \in B \setminus A$. If $s \in A \cap B$, then we can transition from $A$ to $B$ by first removing $p,s$ and then adding $s,q$; otherwise we first add $s,q$ and then remove $p,s$. The case when $q \in A \setminus B$ and $p \in B \setminus A$ follows similarly. In all of these cases, we have that $H$ is exactly the $P$-difference graph where $P$ is the transition path corresponding to the sequence of transitions used.

    Suppose the statement is true for $\ell$ up to $k-1$ and let $H$ be a canonical difference graph with $k$ edge-disjoint paths. Let $A'$ be the state after executing the above procedure on $Q_1$ and $P_1$ be the transition path corresponding to the sequence of transitions from $A$ to $A'$. As argued above, $Q_1$ is the $P_1$-difference graph. Since $H \setminus Q_1$ is a canonical $(A',B)$-difference graph with $k-1$ edge-disjoint paths, we can apply induction to get a transition path $P_2$ from $A'$ to $B$ such that $H \setminus Q_1$ is the $P_2$-difference graph. By concatenating $P_1$ and $P_2$, we get a transition path $P$ from $A$ to $B$. Moreover, $H$ is the $P$-difference graph, as desired.
\end{proof}

 \begin{lemma}  
 \label{lem:canonical}
    Let $A,B$ be states. There exists a shortest path $P$ in the transition graph between $A,B$ such that $\diff(P)$ is  canonical.
 \end{lemma}

 \begin{proof}
    Since $\diff(P)$ is an $(A \symdif B)$-join, it contains a collection of $\ell := |A \symdif B|/2$ edge-disjoint paths $Q_1, \ldots, Q_\ell$ connecting disjoint pairs of requests $p_i,q_i \in A \symdif B$. 

    We now shortcut these paths to produce a  canonical $(A,B)$-difference graph $H$. In particular, for each $i$, if $p_i,q_i \in A \setminus B$ or $p_i,q_i \in B \setminus A$, add the edge $(p_i, q_i)$ to $H$; otherwise, $p_i \in A \setminus B$ and $q_i \in B \setminus A$ (or vice versa), there must be an intermediate node $s_i$ on the path $Q_i$ and we add the edges $(p_i, s_i), (s_i, q_i)$. The existence of $s_i$ is because if $Q_i$ only consists of a single edge $(p_i,q_i)$ then $p_i$ and $q_i$ must be both in $A \setminus B$ or $B \setminus A$, otherwise the transition in $P$ that adds or removes $p_i,q_i$ is invalid.

    By triangle inequality, the total cost of the edges in $H$ is at most that of $Q_1 \cup \cdots \cup Q_\ell$ which in turn is contained in $\diff(P)$. Since $H$ is a canonical $(A,B)$-difference graph, there exists a transition path $P'$ with $\diff(P') = H$. Since $c(P')$ is equal to the total cost of the edges in $\diff(P') = H$, we get the lemma.
 \end{proof}

 \begin{proposition}
    Let $A,B$ be states such that $|A| = |B|+2$. Let $P$ be a shortest path in the transition graph between $A,B$ such that $\diff(P)$ is canonical and $Q_1, \ldots, Q_\ell$ be the corresponding collection of $|A \symdif B|/2$ edge-disjoint paths connecting disjoint pairs of $A \symdif B$. Then, one such path $Q_i$ is a single edge $(p,q)$ with $p,q \in A \setminus B$.
 \end{proposition}
 This proposition follows from the fact that $|A \setminus B| = |B \setminus A| + 2$ and so there must be a path $Q_i$ connecting $p,q \in A \setminus B$. By definition of canonical difference graphs, $Q_i$ is a single edge $(p,q)$.

	We are now ready to formally define \SSCA.
\subparagraph*{Description of \SSCA} 
	Our online algorithm takes as input a sequence of states $\sigma = (S_1,\dots, S_{|\sigma|})$ produced by an online scheduling algorithm. We assume that $\sigma$ satisfies that for all $0 \leq i < |\sigma|-1$, $S_i$ and $S_{i+1}$ are neighbours in $G$. This is without loss of generality: if the input schedule does not satisfy these properties then we can add intermediate timesteps and all states on the shortest path between $S_i$ and $S_{i+1}$ such that it satisfies the above property, and the cost remains the same. When a new state $S_i$ arrives, if $|S_i| \leq |S'_{i-1}|$, the algorithm sets $S'_i = S'_{i-1}$; otherwise, it computes a shortest transition path $P$ between $S_i$ and $S'_{i-1}$  such that $\diff(P)$ is canonical, and sets $S'_i = S'_{i-1} \cup \{p, q\}$ where $p,q \in S_i \setminus S'_{i-1}$ and $(p,q)$ is a path in the path decomposition of $\diff(P)$.

Since \SSCA always transitions to a state that is a superset of its previous state, its solution $\sigma'$ is monotone. 

 Next, we analyse the transition cost of $\sigma'$. 
	\begin{lemma}
		\label{lem:TransitionCost-sizebased}
		$\sum_{i=0}^{|\sigma|-1} c(S^\prime_{i-1}, S^\prime_i) \leq \sum_{i=0}^{|\sigma|-1} c(S_{i-1}, S_i)$.
	\end{lemma}
 
	\begin{proof}
		We prove this lemma by introducing the potential $\phi_i = c(S_i, S'_i)$. 

    We claim that in each iteration $i$, the transition cost of $\sigma'$ is at most the transition cost of $\sigma$ plus the decrease in the potential.

		\begin{claim}
			\label{clm:potential}
			For all $i$, we have $c(S^\prime_{i-1}, S^\prime_i) \leq c(S_{i-1}, S_i) - (\phi_i - \phi_{i-1})$.
		\end{claim}

		\begin{proof}[Proof of Claim~\ref{clm:potential}]
            By triangle inequality, we have
            \begin{equation}
                \label{eq:potential-1}
                c(S_i, S'_{i-1}) \leq c(S_{i-1}, S_i) + c(S_{i-1}, S'_{i-1}).
            \end{equation}

            Next, we will show that
            \begin{equation}
                \label{eq:potential-2}
                c(S'_{i-1}, S'_i) \leq c(S_i, S'_{i-1}) - c(S_i, S'_i).
            \end{equation}  
            Combining these two inequalities yields the claim.
            
            Observe that Inequality~\eqref{eq:potential-2} holds when $S'_i = S'_{i-1}$. Suppose $S'_i \neq S'_{i-1}$. In this case, the algorithm computes the shortest path $P$ between $S_i$ and $S'_{i-1}$ that is canonical and set $S'_i = S'_{i-1} \cup \{p,q\}$ where $p,q \in S_i \setminus S'_{i-1}$ and $(p,q)$ is a path in the path decomposition of $\diff(P)$. We have that $\diff(P) \setminus (p,q)$ is a canonical $(S_i, S'_i)$-difference graph and thus, by Proposition~\ref{prop:realisable}, corresponds to a transition path between $S_i, S'_i$ with length $c(S_i, S'_{i-1}) - c(p,q) = c(S_i, S'_{i-1}) - c(S'_{i-1}, S'_i)$. Thus, $c(S_i, S'_i) \leq c(S_i, S'_{i-1}) - c(S'_{i-1}, S'_i)$. Rearranging this inequality yields Inequality~\eqref{eq:potential-2}. This completes the proof of the claim.
		\end{proof}
		Using Claim~\ref{clm:potential}, we determine the total transition cost incurred by $\sigma^\prime$ as follows.
		\begin{align*}
			\sum_{i=1}^{|\sigma|} c(S^\prime_{i-1}, S^\prime_i) 
            &\leq \sum_{i=1}^{|\sigma|} c(S_{i-1}, S_i) - \sum_{i=1}^{|\sigma|} (\phi_i - \phi_{i-1})\\
			&= \sum_{i=1}^{|\sigma|} c(S_{i-1}, S_i) - \phi_{|\sigma|} + \phi_{0}\\
			&\leq \sum_{i=1}^{|\sigma|} c(S_{i-1}, S_i).
		\end{align*}
		The last inequality holds because $\phi_{0} = 0$ and the potential is non-negative.
	\end{proof}

	The cost of an \MPMGADMTS solution consists of the transitions cost, as well as the \ProcessCost. Since in every iteration $i$, we have $|S'_i| \geq |S_i|$ and the processing cost of a state is a monotone non-increasing function of the size of the state, we get that the total processing cost of $\sigma'$ is at most that of $\sigma$. Together with Lemma~\ref{lem:TransitionCost-sizebased}, we get that the total cost of $\sigma'$ is at most that of $\sigma$. This concludes the proof of Lemma~\ref{lem:any_to_monotone}.
\end{proof} \subsection{Applying MTS Algorithms to \MPMGADMTS}
\label{ssec:analysis}
In this section, we prove Corollary~\ref{cor:size-based-to-match}. Consider an instance of \MPMGAD with $m$ requests in a metric space of $n$ points and the instance of \MPMGADMTS created by applying Theorem~\ref{thm:size-based-to-match}. Let $N$ be the number of states of the \MPMGADMTS instance.

There are two issues that arise when applying MTS algorithms to \MPMGADMTS directly. 
\subsubsection{Eliminating the Diameter}
The first issue is that all known MTS algorithms have a cost bound of the form $f(N) \cdot cost(OPT) + D$ where $OPT$ is the optimal MTS solution and $D$ is the diameter of the MTS state space. Observe that $D$ is at least the distance between the empty matching and the max-cost perfect matching, i.e.~the cost of the max-cost perfect matching. Unfortunately, the cost of the max-cost perfect matching can be much larger than that of the optimal solution. To overcome this, one could restrict the MTS solution to only use states whose distance from the initial state is at most $cost(OPT)$. This can be achieved by setting the costs of the other states to be infinite. This effectively reduces the diameter of the state space to at most $2 \cdot cost(OPT)$, and would give us a cost bound of $O(f(N)) \cdot cost(OPT) + 2 \cdot cost(OPT)$. 
The issue is that, since the MTS tasks arrive in an online fashion, the optimal solution remains unknown until all tasks have arrived. To address this issue we use the Guess-and-Double Method, which, maintains a guess of the value of the optimal solution as the tasks are processed. This guess is used to determine the diameter of the state space used by the algorithm. When the guess becomes too small, the value of the guess is increased and the algorithm is simulated on the input that has already been processed, only this time on the larger state space, to determine the state it would now be in, and processes the new tasks accordingly. For the benefit of the reader, we give a high level overview of the method below.
\subparagraph{The Guess-and-Double Method} Let $R = [r_1, r_2, ... r_T]$ be the sequence of tasks given by the MTS problem. Let $OPT_t$ be the cost of the optimal solution for processing the tasks that have arrived up to and including time $t$. At any time $t$ we maintain a \emph{guess} $j$ such that

\begin{align}
    \label{eq:constraint}
    2^{j-1} < OPT_t \leq 2^{j}.
\end{align}

When our guess $j$ no longer satisfies~\eqref{eq:constraint} we increase it's value (thus doubling the radius of the metric space) until it is back within our bounds. Each new guess instantiates a new \emph{phase}. At the start of each phase, we \emph{restart} the MTS algorithm with a superset of the previous state space, which now consists of all states within distance $2^{j}$ from the original start state. Note that by \emph{restart} we mean that we simulate the MTS algorithm with the new diameter on all tasks that have arrived so far, and process the new tasks in accordance with the decisions made by the algorithm with the new diameter.

Our initial guess $j$ satisfies $2^{j-1} < OPT_1 \leq 2^{j}$. We define $MTS_j$ to be the MTS algorithm that operates on the given state space with diameter $2^j$. Let $R_t = [r_1,...,r_t]$ be the sequence of tasks that have arrived up to and including timestep $t$. We define $MTS_j(R_t)$ to be the state $MTS_j$ would end up in after processing $R_t$. The Guess-and-Double method, for every timestep $t$, maintains a guess $j$ that satisfies~\eqref{eq:constraint}, and moves to $MTS_j(R_t)$. \\

\begin{algorithm}[H]
	\SetAlgoLined
    Initialise the first phase.\\
    Choose $j$ such that $2^{j-1} < OPT_1 \leq 2^{j}$.\\
    \For{Every timestep $t$}{
        \If{$2^{j-1} < OPT_t \leq 2^{j}$}{
            Move to state $MTS_j(R_t)$
        }
        \Else{
            End the previous phase and initialise a new phase.\\
            Update the value of $j$ such that $2^{j-1} < OPT_t \leq 2^{j}$.\\
            Move to state $MTS_j(R_t)$
        }
    }
	\caption{Updating the guess and splitting the tasks into phases.}
\end{algorithm}
\leavevmode\newline 

Note that each time we update the value of our guess, the value of the optimal solution has at least doubled since we made our last guess. 

To analyse the total cost of the resulting schedule we look at two separate costs. The first is the cost incurred within each phase. This includes the transition costs incurred during the phase (from moving between states), plus the cost associated with servicing all tasks that arrived during the phase in the states the algorithm moved through during the phase. We refer to this cost as the \emph{internal} phase cost.
The second cost consist of the transition costs incurred in moving between the last state of a phase $i$, and the first state of the consecutive phase $i+1$.  We refer to this cost as the \emph{external} phase cost.

We start by bounding the internal phase cost. Let $cost(MTS_j)$ denote the internal phase cost of the phase associated with guess $j$. Because $cost(MTS_j)$ can be upper bounded by the cost the algorithm would have incurred for processing all tasks that arrived prior to and during the phase associated with guess $j$, we can bound $cost(MTS_j)$ as follows:
\begin{align*}
    cost(MTS_j) \leq (O(f(N)) + 2) \cdot cost(OPT_j).
\end{align*}
Therefore, the total internal phase cost incurred over all $k$ phases is 
\begin{align*}
    \sum_{i=1}^k ALG_i &\leq (O(f(N)) + 2) \cdot \sum_{i=1}^k cost(OPT_i)\\
    & \leq 2 \cdot (O(f(N)) + 2) \cdot cost(OPT_k).
\end{align*}

Next, we bound the external phase cost. Let $S_i$ be the last state of a given phase $i$, associated with a guess $j$, and let $S_{i+1}^\prime$ be the first state of phase $i+1$, associated with the next guess $j^\prime$. We bound $c(S_i, S_{i+1}^\prime)$\footnote{Recall that we denote the transition cost of going from state $S_i$ to $S_{i+1}^\prime$ by $c(S_i, S_{i+1}^\prime)$.} as follows:

Let $S_0$ be the empty start state.
\begin{align*}
    c(S_i,S_{i+1}^\prime) \leq c(S_i,S_0) + c(S_{i+1}^\prime,S_0).
\end{align*}
For all consecutive phases $i$ and $i+1$, associated with guesses $j$ and $j^\prime$ respectively, it holds that
\begin{align*}
    c(S_i,S_0) \leq (O(f(N)) + 2) \cdot cost(OPT_i)
\end{align*}
and 
\begin{align*}
    c(S_{i+1}^\prime,S_0) \leq (O(f(N)) + 2) \cdot cost(OPT_{i+1}).
\end{align*}
We thus bound the cost over all $k-1$ phase transitions as follows
\begin{align*}
    \sum_{i = 1}^{k-1}c(S_i,S_{i+1}^\prime) &\leq 2 \cdot \sum_{i = 1}^{k-1} (O(f(N)) + 2) \cdot cost(OPT_{i+1})\\
    &\leq 4 \cdot (O(f(N)) + 2) \cdot cost(OPT_k).
\end{align*}
The total cost of the solution produced by the Guess-and-Double method can thus be bounded by $6 \cdot (O(f(N)) + 2) \cdot cost(OPT_k)$.

We can now use the $O(N)$-competitive deterministic algorithm of~\cite{originalMTS} to obtain our deterministic algorithm for MPMD-Size.

\subsubsection{The Need for an Online Embedding}
The second issue stems from the fact that the reduction in Theorem~\ref{thm:size-based-to-match} creates an MTS instance where the states are arriving over time. This is because the states correspond to matchings of requests and the requests are arriving online. This does not pose a problem for the deterministic $O(N)$-competitive Work Function Algorithm of~\cite{originalMTS}. However, we cannot directly apply the current-best randomised algorithm for MTS of~\cite{randomisedMTS} as it pre-computes a probabilistic embedding of the MTS metric space into a hierarchically separated tree (HST). Instead, we need to use a probabilistic online embedding into a HST together with the $O(\log N)$-competitive randomized algorithm for MTS on HSTs of~\cite{randomisedMTS}. Using the online embedding of~\cite{onlineEmbeddingIndyk} adds a factor of $O(\log N \log \Phi)$ where $\Phi$ is the ratio of the largest distance to the smallest distance in the MTS state space, i.e.~the aspect ratio. However, $\Phi$ can be arbitrarily large. We deal with this by proving that the Abstract Network Design Framework of~\cite{onlineEmbedding} can be extended to apply to MTS. For the benefit of the reader, we give a short overview of the Abstract Network Design Framework. 

\subparagraph{The Abstract Network Design Framework}
In an instance of abstract network design, the algorithm is given a connected graph $G(V,E)$ with edge lengths $d: E \rightarrow \mathbb{R}_+$. At each timestep $i$, a requests that consists of a set of points in the graph, called \emph{terminals}, arrive in an online fashion. The algorithm provides a response $R_i = (G_i, C_i)$, which consists of a subgraph $G_i \subseteq E$, and a connectivity list $C_i$, which is an ordered subset of terminal pairs from the terminals that have arrived so far, and determines what will become (and remain) connected from this timestep onwards. The algorithm is given, at each timestep, a feasibility function $F_i: (C_1, ..., C_i) \rightarrow \{0,1\}$ that maps a sequence of connectivity lists to either 0 (infeasible) or 1 (feasible). A solution $S_i$, which consists of a sequence of responses for every time step up to and including timestep $i$, is feasible if $F_i(C_1, ..., C_i) = 1$ and all pairs in $C_j$ are connected in all $G_i$ for all $i \geq j$. To determine the cost of a solution to the first $i$ requests $S_i$, the framework uses a load function $\rho_i: 2^{\{1,...,i\}} \rightarrow \mathbb{R}_+$, which takes as input the sequence of timesteps in which the edge was used, and outputs a corresponding multiplier to the cost of an edge $d(e)$. It aims to model how the cost of using an edge grows as the edge is used multiple times. The function must be subadditive, monotone non-decreasing, and satisfy $\rho_i(I) = 0$ if and only if $I = \emptyset$. The total cost of a solution $S_i$ is defined as follows.
\begin{align*}
    cost(S_i) = \sum_{e \in E} d(e) \cdot \rho(\{j \leq i: e \in G_j\})
\end{align*}
\subparagraph{Extending the Abstract Network Design Framework} Though we can express the transition cost of an instance of general MTS using this framework, we cannot express the cost vectors associated with processing the tasks in MTS in the current state of the framework. In order to address this issue, we propose the following alterations to generalise the framework.

We replace the feasibility function with a function $F^\prime_i: (C_1, ..., C_i) \rightarrow \mathbb{R_+}$, where $F^\prime_i(C_1, ..., C_i)$ is the processing cost of the algorithm during timestep $i$ if the solution $S_i = ((G_1, C_1),\dots,(G_i, C_i))$ is feasible, and $\infty$ otherwise.

We now re-define the cost of a solution to the first $i$ requests $S_i$ to incorporate the total processing cost incurred by the algorithm. 
\begin{align*}
    cost(S_i) = \sum_{e \in E} d(e) \cdot \rho(\{j \leq i: e \in G_j\} + \sum_{l=1}^i F^\prime(C_1, \ldots, C_l)
\end{align*}
Since the processing cost provided to the algorithm is independent of the metric space, it follows that the processing cost remains unaffected by the online embedding. It therefore does not affect the overhead due to the online embedding.

Note that the extension of this framework means it can now be used to model online problems with delay, where the delay cost can be modelled as the processing cost.

\subparagraph{Expressing MTS in the Abstract Network Design Framework} It remains to formulate the general MTS problem in the Extended Abstract Network Design Framework defined above. 

We define the terminal set of the $i$th request to be the set of states that have arrived so far. We define the cost of an edge $d((u,v))$ to be the transition cost between states $u,v$. Let $\mathbb{T}_i$ be the cost vector associated with processing task $i$, and $T_i(w)$ be the cost of processing task $i$ in state $w$. Let $v$ be the last terminal in $C_i$. The extended feasibility function $F_i$ is defined by $F_i(C_1,\ldots,C_i) = T_i(v)$ if there is only a single ordered pair in each $C_j$ for all $j \leq i$ and the sequence of $(C_1,\ldots, C_i)$ is a valid path. The load function is simply the cardinality function because we pay the transition cost associated with the edge each time we transition to a different state.  
\subparagraph{Min-operator} Bartal et al.~\cite{onlineEmbedding} show that if the problem can be captured by the Abstract Network Design Framework and admits a \emph{min-operator}, it is possible to reduce the overhead due to the online embedding to $O(\log^3 N)$. 
\begin{definition}[Min-operator]
An algorithm admits a min-operator  with factor $\mu \geq 1$ if there exists a competitive algorithm\footnote{This algorithm may be randomised.} for the problem, and for any two deterministic online algorithms $A$ and $B$\footnote{Note that these algorithms need not be competitive on all instances of the problem.}, there exists a third online deterministic algorithm $C$ such that the cost of $C$ satisfies $cost(C) \leq \mu \cdot \min\{cost(A), cost(B)\}$, where $cost(A)$ and $cost(B)$ are the respective costs of algorithms $A$ and $B$. If either algorithms $A$ or $B$ are randomised, the expected cost of $C$ must satisfy $E[cost(C)] \leq \mu \cdot \min\{E[cost(A)], E[cost(B)]\}$.
\end{definition}

We have shown that the results by Bartal et al.~\cite{onlineEmbedding} also hold for the extended Abstract Network Design Framework. Since the MTS problem in general admits a min-operator~\cite{MTS_min_operator}, using the framework of~\cite{onlineEmbedding} allows us to reduce the overhead due to the online embedding to $O(\log^3 N)$ for an overall competitive ratio of $O(\log^4 N)$. 
\begin{toappendix}
\section{A Deterministic Lower Bound for \MPMDSize} \label{subs:deterministic_lb}
In this section we lower bound the competitive ratio of any deterministic online matching algorithms for \MPMDSize.
\begin{theorem}
	 Every deterministic algorithm for \MPMDSize has competitive ratio $\Omega(n)$, where $n$ is the number of points in the metric space. 
\end{theorem}
\begin{proof}
Consider an $n$-point uniform metric space with distance 1 between all points. We fix a deterministic online matching algorithm ALG that will process a request sequence of size $2n-2$ determined by an adversary. The aim is to force ALG to match requests at two distinct points in the metric space $n-1$ times by ensuring that each time ALG needs to match two requests there is at most a single unmatched request available at each point in the metric space. We then ensure the optimal solution to the instance is to match requests at distinct points only once by placing two requests at $n-2$ points in the metric space, and placing only a single request at two points in the metric space. To this end we define the adversary to satisfy the following properties at all times:

\begin{enumerate}
	\item A new request is only ever placed at a point with no unmatched requests. 
	\item Every point in the metric space receives at most 2 requests.
\end{enumerate}

The first property of the adversary ensures that at most one unmatched request is available at any point in the metric space at all times and thereby aims to ensure ALG must match requests across two distinct points each time. The second property aims to ensure the optimal solution is to only match requests at distinct points a single time.

Before we define the behaviour of the adversary in more detail, we divide time up into phases and define the delay function in relation to each phase. The first phase starts at $t=0$ and ends when ALG performs a match. The next phase begins when the previous phase ends, and the last phase ends when ALG has matched the last request in the request sequence. The delay function is the same for all timesteps $t$ within the same phase. For the $i$th phase the delay function is defined as follows.\\

For all timesteps $t$ in phase $i$, for any subset of requests $S$,
\begin{equation}
	f_t(S) = 
	\begin{cases} 
	      0 & |S| \leq (n-i)\\
	      \infty & |S| = (n-i)+1\\
	\end{cases}
\end{equation}

We are now ready to define the behaviour of the adversary in more detail. At time $t=0$, the first phase begins and the adversary places $n$ requests, one at every point in the metric space. This satisfies all properties of the adversary and invokes an infinite delay cost for ALG, forcing the algorithm to perform a match. Since we have exactly 1 unmatched request at every point, ALG is forced to match requests at distinct points in the metric space and incur a distance cost of 1.

To define the behaviour of the adversary during the remaining phases, we first introduce the following terminology.
\begin{definition}[Saturated points]
	We call a point in the metric space \emph{saturated} if it has received at least 2 requests so far, and \emph{unsaturated} otherwise. 
\end{definition}

\begin{definition}[Active points]
	We call a point in the metric space \emph{active} if it currently hosts an unmatched request, and \emph{inactive} otherwise. 
\end{definition}

For all $i \in \{2... n-1\}$, at the start of phase $i$, the adversary will place a request at an inactive unsaturated point. By definition, this satisfies the properties of the adversary and thereby forces ALG to match two distinct points in the metric space to avoid infinite delay.

Before we analyse the competitive ratio we need to show that the adversary is well-defined. To this end we prove that, regardless of the behaviour of ALG, at the end of each phase $i \in \{1... n-2\}$, after ALG has performed a match, there always exists an inactive unsaturated point that the adversary can place a request on. We formalise this in the following claim.

\begin{claim}
	\label{clm:adversary-properties}
	At the end of each phase $i$ (after ALG has matched), there exist $i+1$ inactive points and at least two of these points are unsaturated. 
\end{claim}

\begin{proof}[Proof of Claim~\ref{clm:adversary-properties}]
We prove the claim by induction on the number of phases $i$. 

The base case is when $i = 1$. At the start of phase 1, every point in the metric space receives a single request. At the end of phase 1, ALG is forced to match 2 arbitrary requests at different points in the metric space. Therefore, at the end of phase 1, there will be 2 points in the metric space that are both inactive and unsaturated. 

Assume the claim holds for all phases $i = 1$ up to $i = k$. By this assumption, at the start of phase $k+1$ we have $k+1$ points in the metric space that are inactive and at least two of these requests are unsaturated. The adversary can now place a new requests at one of the two unsaturated points, satisfying its current properties, and leaving $k$ points inactive (one of which is still unsaturated). At the end of phase $k+1$, ALG will be forced to match requests across two distinct points in the metric space again, adding another 2 points to the set of inactive points. By the end of phase $k+1$, there are thus a total of $k+2$ inactive points and at least one of them is unsaturated. It now remains to prove at least two of the inactive points must be unsaturated. Assume for the sake of contradiction that this is not the case. Then all remaining $k+1$ inactive points must be saturated. Since the points are inactive and saturated this means that ALG has matched at least $2\cdot(k+1)$ points in $k+1$ phases. In every phase ALG matches exactly 2 points. By the end of phase $k+1$, ALG can thus have matched at most $2\cdot(k+1)$ requests. But the $k+1$ saturated and 1 unsaturated points imply at least $2\cdot(k+1)+1$ requests must have been matched by the end of this phase. This constitutes a contradiction. We conclude that there are $k+2$ inactive points by the end of iteration $k+1$ and that at least $2$ of them are unsaturated. By the principle of induction, Claim~\ref{clm:adversary-properties} holds.
\end{proof}

We conclude the adversary is well-defined. It remains to analyse the competitive ratio. Because the adversary maintained the first property, it follows that each point in the metric space hosts at most one unmatched request at any time. From this we conclude that ALG incurred a distance cost of 1 during each phase, resulting in a total distance cost of $n-1$. Furthermore, because the adversary maintained the second property and a total of $2n-2$ request arrived at $n$ points in the metric space, it follows that two requests arrived at $n-2$ points in the metric space and one request arrived at two points in the metric space. Let us refer to the latter two points as $v_1$ and $v_n$. The optimal solution is to match the requests at $v_1$ and $v_n$ in the first phase and in each consecutive phase, to match the two requests at the same point. Since neither ALG nor OPT incurred any delay cost, the total cost of ALG and OPT are as follows:
\begin{align*}
	&Cost(OPT) = 1\\
	&Cost(ALG) = n-1
\end{align*}
From this we conclude our lower bound on the competitive ratio for any deterministic online matching algorithm on \MPMDSize.
\end{proof}
\end{toappendix} 

\begin{toappendix}
\section{A Randomised Lower Bound for \MPMGAD with Size-based Delay}
In this section, we lower bound the competitive ratio of any randomised online matching algorithm for \MPMDSize.
\begin{theorem}
	\label{thm:randomised_lb}
	Every randomised algorithm for \MPMDSize has competitive ratio $\Omega(\log n)$, where $n$ is the number of points in the metric space.
\end{theorem}

\begin{proof}
Consider an $n$-point uniform metric space with distance 1 between all points. Applying Yao’s principle~\cite{YaoMinimax}, we define a uniform random distribution over the inputs such that any deterministic online matching algorithm will have expected cost $\Omega(\log n)$. We define the behaviour of the adversary to ensure the expected cost of the optimal solution is always 1. To this end we define the adversary, which will place $2n-2$ requests, to satisfy the following property at all times:

\begin{itemize}
	\item Every point in the metric space receives at most 2 requests.
\end{itemize}

This property ensures the optimal solution is to only match requests at distinct points a single time, regardless of where the requests are placed in the metric space.

Similar to the previous section, we divide time up into phases and define the delay function in relation to each phase. The first phase starts at $t=0$ and ends when ALG performs a match. The next phase begins when the previous phase ends, and the last phase ends when ALG has matched the last request in the request sequence. The delay function is the same for all timesteps $t$ within the same phase. For the $i$th phase the delay function is defined as follows.\\

For all timesteps $t$ in phase $i$, for any subset of requests $S$,
\begin{equation}
	f_t(S) = 
	\begin{cases} 
	      0 & |S| \leq (n-i)\\
	      \infty & |S| = (n-i)+1\\
	\end{cases}
\end{equation}

We now define the behaviour of the adversary in more detail. At time $t=0$, the first phase begins and the adversary places $n$ requests, one at every point in the metric space. This satisfies all properties of the adversary and invokes an infinite delay cost for ALG, forcing the algorithm to perform a match. Since we have exactly 1 unmatched request at every point, ALG is forced to match two requests at distinct points in the metric space, and incur a distance cost of 1.
In order to define the behaviour of the adversary on the remaining phases, recall from Section~\ref{subs:deterministic_lb} that we define a point to be \emph{saturated} if it has received 2 requests (and \emph{unsaturated} otherwise) and \emph{active} if it hosts an unmatched request (and \emph{inactive} otherwise). 

For all $i \in \{2... n-1\}$, at the start of phase $i$, the adversary will drop a request uniformly at random at any of the \emph{unsaturated} points. By definition, this satisfies the adversary property. 

We observe that the expected cost of any deterministic algorithm ALG on the input sequence described above is lower bounded by the expected number of phases in which it is forced to match two requests at distinct points in the metric space.

We thus wish to compute, for each phase $i$, a lower bound on the probability that ALG will have to match two distinct points and thereby incur a distance cost of 1. We note that if the adversary drops a request on an \emph{active} unsaturated point, ALG can match the two requests at that point without incurring any distance cost. We conclude that ALG will only be \emph{forced} to match requests at two distinct points in the metric space if the adversary places the request on an \emph{inactive} unsaturated point. 
It would thus be helpful to consider the number of unsaturated points the adversary can place a request at during phase $i$, as well as how many of these are inactive. 
To this end we state the following observation followed by a claim regarding the number of unsaturated points, as well as the number of inactive unsaturated points at the end of phase $i$. (Note that Claim~\ref{clm:adversary-properties-rando} is the same claim as we made in Section~\ref{subs:deterministic_lb}, which requires a new proof due to the change in behaviour of the adversary).

\begin{observation}
    \label{obs:nr_unsaturated}
At the end of every phase $i$ (after ALG has matched), there exist $n-i+1$ unsaturated points. 
\end{observation}
\begin{claim}
	\label{clm:adversary-properties-rando}
At the end of every phase $i$ (after ALG has matched), there exist $i+1$ inactive points and at least two of these points are unsaturated. 
\end{claim}

\begin{proof}[Proof of Claim~\ref{clm:adversary-properties-rando}]
We prove the claim by induction on the number of phases $i$. 

The base case is when $i = 1$. At the start of phase 1, every point in the metric space receives a single request. At the end of phase 1, ALG is forced to match 2 arbitrary requests at different points in the metric space. Therefore, at the end of phase 1, there will be 2 points in the metric space that are both inactive and unsaturated. 

Assume the claim holds for all phases $i = 1$ up to $i = k$. Thus, at the start of phase $k+1$, we have $k+1$ points in the metric space that are inactive and at least two of these requests are unsaturated. The adversary can now place a new requests at any of the unsaturated points (both active and inactive). 

If it places it on an active unsaturated point, ALG is able to match the two active requests at the same point and incur no distance cost. However, the point will then become inactive (while it was previously active) and thus add 1 to the count of inactive points. This results in $k+2$ inactive points. Furthermore, since the adversary did not place a request at any of the inactive unsaturated points, we still have at least 2 inactive unsaturated points. We conclude that by the end of phase $k+1$ we have $k+2$ inactive points and at least two of them must be unsaturated. 

On the other hand, if the adversary manages to place the request on an inactive unsaturated point, this leaves $k$ points inactive (one of which is still unsaturated). At the end of phase $k+1$, ALG will be forced to match requests across two distinct points in the metric space again, adding another 2 points to the set of inactive points. By the end of phase $k+1$, there are thus a total of $k+2$ inactive points and at least one of them is unsaturated. It now remains to prove at least two of the inactive points must be unsaturated. Assume for the sake of contradiction that this is not the case. Then all remaining $k+1$ inactive points must be saturated. Since the points are inactive and saturated this means that ALG has matched at least $2\cdot(k+1)$ points in $k+1$ phases. In every phase ALG matches exactly 2 points. By the end of phase $k+1$, ALG can thus have matched at most $2\cdot(k+1)$ requests. But the $k+1$ saturated and 1 unsaturated points imply at least $2\cdot(k+1)+1$ requests must have been matched by the end of this phase. This constitutes a contradiction. We conclude that there are $k+2$ inactive points by the end of iteration $k+1$ and that at least $2$ of them are unsaturated. By the principle of induction, Claim~\ref{clm:adversary-properties-rando} holds.
\end{proof}

The claim implies that at the end of phase $i$, there exist $n-i+1$ unsaturated points which the adversary can place a new request at, and at least two of those points are inactive. If the adversary manages to place a request at one of the unsaturated inactive points, this will force ALG to match two distinct points in the metric space and hence incur a distance cost. We define the random variable $b_i$ to be the number of inactive points among the unsaturated points. Conditioned on $b_i$, the probability that ALG will need to match across in phase $i$ can be expressed as $\frac{b_i}{n-i+1}$. We now use this to lower bound the total expected cost of ALG. Let $X_i$ be the indicator variable that the algorithm pays a distance cost of 1 to match requests at distinct points in phase $i$.

\begin{align*}
	E[cost(Alg)] &= \sum_{i=1}^{n-1} E[X_i]\\
    &= \sum_{i=1}^{n-1} E[E[X_i | b_i]]\\
	&\geq \sum_{i=1}^{n-1} \frac{E[b_i]}{n-i+1}\\
	&\geq \sum_{i=1}^{n-1} \frac{2}{n-i+1}\\
	&= 2\sum_{k=2}^{n} \frac{1}{k}\\
	&= \Omega(\log n).
\end{align*}
The second inequality follows from Claim~\ref{clm:adversary-properties-rando}, and the last from harmonic series. 

Because the adversary only places at most 2 requests at each of the points and places $2n-2$ requests in total for all inputs in the distribution, this will result in $n-2$ points that receive exactly 2 requests and 2 points that receive exactly 1 requests. The optimal offline solution is to match the single requests at the two points that received a single request, and to match all other requests locally at the same point. The expected cost of the optimal solution over the distribution of inputs is thus:
\begin{align*}
	E[cost(OPT)] = 1.
\end{align*}

From this we conclude the correctness of Theorem~\ref{thm:randomised_lb}.
\end{proof}
\end{toappendix}  

\section{Online Matching with Concave Delay}
    \label{sec:concave}
In this section we aim to demonstrate that the framework underlying some existing techniques, applied by Bienkowski et al.~\cite{pd}, can be modified to design deterministic competitive algorithms for MPMD with concave delays. Furthermore, our techniques work for any delay function that can extend a metric space into a time-augmented metric space (see Section~\ref{sec:notation} for details on a time-augmented metric space).

\subsection{Notations and Preliminaries}
\label{sec:notation}
In this section we introduce basic notation that will be used throughout the chapter. We then explain the (offline) moat-growing framework, introduced by Goemans and Williamson~\cite{forest}, and address how this framework can be moved to the online setting with delay. We also explain the concept of a time-augmented metric space, and extend it by introducing the concept of a concave time-augmented metric space. 

Recall that we define MPMD on a metric space $(V,\dist)$ where $V$ is a set of $n$ points. A sequence $R$ of $m$ requests arrive over time in the metric space. We assume without loss of generality that the requests arrive at unique points in the metric space. 
In this section, we partially follow the notation of Bienkowski et al.~\cite{pd}. Let $R_t$ be the set of requests that have arrived in the metric space up to time $t$, and $S \subseteq R_t$ be a subset of these requests. Let $\Sur(S)$ denote the \emph{surplus} of $S$, which is equal to 1 if the cardinality of the set $S$ is odd, and 0 otherwise. Furthermore, let $\delta(S) = \{(u,v): u \in S, v \notin S\}$ be the cut of $S$. For each request $u$, we denote its position in the metric space as $\pos(u)$, and its arrival time as $\atime(u)$. 
Let $E = {R \choose 2}$ be the set of all possible undirected pairs $u,v \in R$. Each pair $u,v \in R$ is represented by an edge $e = (u,v)$, $e\in E$. 

If the reader is familiar with the Moat-Growing Framework by Goemans and Williamson~\cite{forest}, they can skip to Section~\ref{par:introDelay}.

\subsubsection{The Moat-Growing Framework}
For the benefit of the reader, we give a short high-level overview of the moat-growing framework, by Goemans and Williamson~\cite{forest},
upon which we base our algorithm and analysis. We explain the framework as applied to the offline minimum cost perfect matching problem in~\cite{forest}. Define the metric min-cost perfect matching problem on a graph $G = (R,E)$ whose edge costs, defined by $c: E \rightarrow \mathbb{Q}_+$, satisfy triangle inequality. Define a function $f: 2^{R} \rightarrow \{0,1\}$, where f(S) = 1 if the cardinality of the set is odd, and 0 otherwise. The problem is first formulated as an integer program (IP), where the objective function is a minimisation function of the total cost associated with the chosen edges. A constraint, representing a connectivity requirement, is associated with each subset of vertices. The IP is as follows:
\begin{alignat*}{2}
  & \text{minimize: } & & \sum_{e\in E} c_e \cdot x_e \\
   & \text{subject to: }& \quad & \sum_{\mathclap{{e \in \delta(S)}}}
                \begin{aligned}[t]
                    x_e & \geq f(S), &\qquad \quad \forall \emptyset \neq S \subset R\\[3ex]
                  x_{e} & \in \{0,1\}, &\qquad \quad \forall e \in E
                \end{aligned}
\end{alignat*}
Relaxing the IP into a linear program (LP) gives us a lower bound on the cost of the optimal solution to the IP. By weak duality, any dual solution lower bounds the primal solution, and hence, the optimal primal IP solution. The dual of the LP is as follows:
\begin{alignat*}{2}
  & \text{maximise: } & & \sum_{S \subseteq R} y_S \cdot f(S) \\
   & \text{subject to: }& \quad & \sum_{\mathclap{S:e \in \delta(S)}}
                \begin{aligned}[t]
                    y_S & \leq c_e, &\qquad \quad \forall e \in E\\[3ex]
                  y_S & \geq 0,	&\qquad \quad \forall \emptyset \neq S \subset R
                \end{aligned}
\end{alignat*}
The algorithm maintains a family of active sets, corresponding to primal constraints not yet satisfied. For the sake of the analysis, the algorithm maintains a feasible dual solution while constructing a primal solution to the IP. It does so by continuously increasing the dual variables associated with the active sets until at least one of the constraints is \emph{tight}. A constraints is said to be \emph{tight} if it holds with equality. A dual variable associated with a set of requests $S$ can be visualised as a moat surrounding the set in the metric space. As we increase the dual variables, the moat surrounding the set grows. When a dual constraint becomes tight it can be visualised as the boundary of two moats around the sets $S_1$ and $S_2$ touching and, consequently, the accumulation of all moats around the sets containing the two endpoints will have covered the edge (this process is visualised in Figure~\ref{fig:GrowingDual}). When this happens, the algorithm adds the corresponding edge to the primal IP solution, deactivates the two sets, and, if $f(S_1 \cup S_2) = 1$, activates the set that is the union of the previous two sets $S_3 = S_1 \cup S_2$. By deactivating the previous two sets, the algorithm avoids violating any dual constraints. The primal-dual part of the algorithm terminates when all primal constraints have been satisfied, resulting in a feasible primal IP solution. 
\begin{figure}[H]
	\begin{center}
		\includegraphics[scale=0.3]{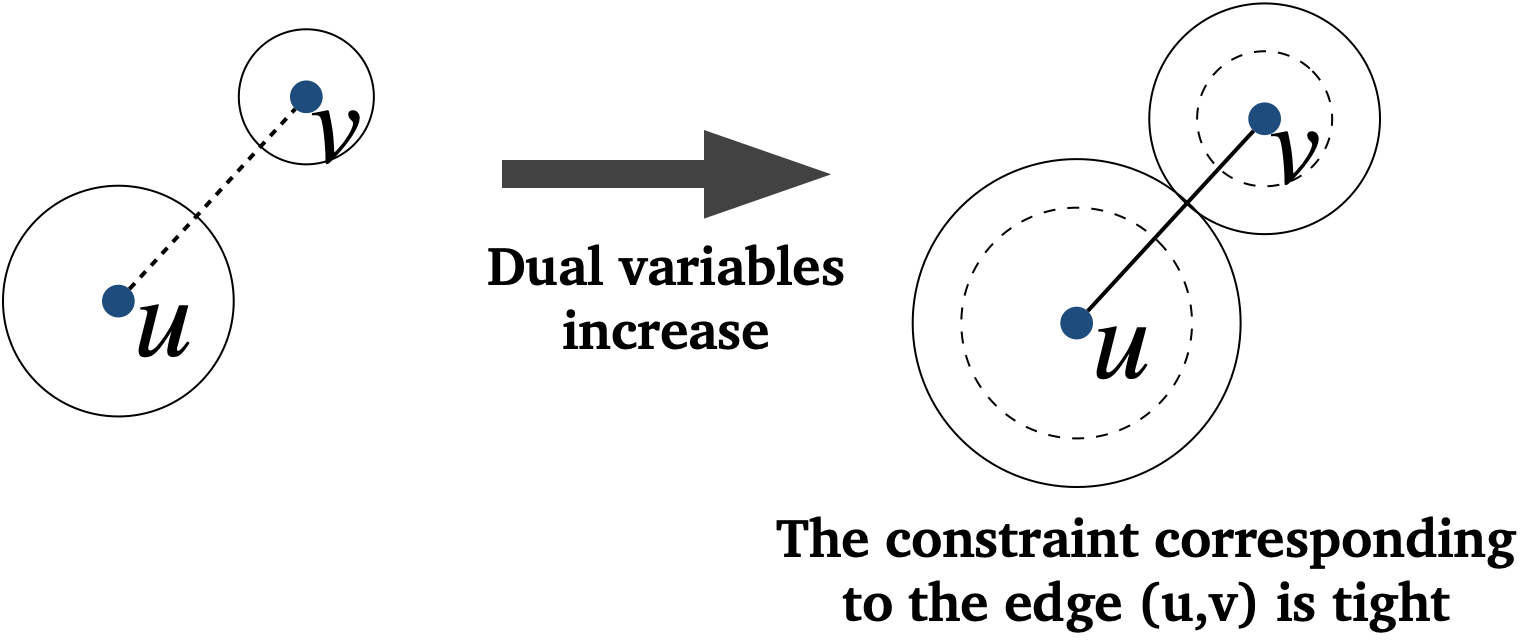}
		\caption{A visualisation of a tight constraint event.}
		\label{fig:GrowingDual}
	\end{center}
\end{figure}
The IP solution is now a forest, where every connected component consists of an even number of requests and the degree of every vertex inside each component is odd. For each connected component, the algorithm then performs a Eulerian Tour of the component and skips each vertex that has already been visited. The resulting cycle is at most twice the cost of the original component. As the cycle is now the union of two perfect matchings of the requests inside the component, the algorithm chooses the one with the smallest cost. The resulting matching will have cost at most half the cost of the Eulerian tour, which has cost at most twice that of the component, whose cost is paid for by the dual solution.

Upon termination of the algorithm, since any feasible dual solution can be seen as a lower bound to the cost of an optimal primal LP solution, and the primal LP solution is a lower bound to the optimal IP solution, one can bound the competitive ratio of the algorithm by relating the resulting IP solution to the value of the dual solution produced. 

\paragraph{Introducing Delay to the Framework}
\label{par:introDelay}When introducing uniform linear delay to the problem, the progression of time becomes important and the framework requires a way to account for the accumulation of delay cost. To this end, one can define an instance where the cost of an edge is the sum of the distance between the requests in the original graph plus the delay cost the first request would incur from the time it arrives until the second request arrives. This only works because an optimal offline solution will always match a pair of requests upon the arrival of the second request. By defining the instance in this way, the underlying metric space incorporates both the distance cost in the original metric space as well as the delay cost incurred by a matched pair of requests. The approach can be interpreted as moving the problem to a \emph{time-augmented metric space}. This idea is used implicitly, not only by Bienkowski et al.~\cite{pd}, but also in all other deterministic competitive algorithms for MPMD with uniform linear delay, such as~\cite{hemispheres} (who explicitly use a time-augmented metric space) and \cite{spheres}.

\begin{definition}[Time-augmented metric space]
Given a metric space $D = (W,\dist)$ we define the \emph{time-augmented metric space} as $D_T = (W \times \mathbb{R},\timedist)$ where $\timedist$ is a distance function defined as 
\begin{align*}
	\timedist(u,v) = \dist(\pos(u), \pos(v)) + |\atime(u) - \atime(v)|.
\end{align*}
Where $u,v \in W \times \mathbb{R}$. In other words, a time-augmented metric space is a metric that is the Cartesian product of a metric space $D$ and the time axis, such that the distance between two points in the original metric space is defined as the difference in position in the original metric space, in addition to the difference in arrival times.
\end{definition}
One can easily verify that a time-augmented metric space, by the above definition, is indeed a valid metric space that satisfies triangle inequality. 

\paragraph{Moving the Framework to the Online Setting}Finally, to move this framework to the online setting, more challenges need to be overcome. 

Firstly, in the online setting, decisions cannot be revoked and hence, the final pruning of the connected component into a valid matching is no longer a possibility. In response to this issue, Bienkowski et al.~\cite{pd} only \emph{mark} the edges corresponding to a tight constraint and, in the event of a tight constraint, adds to the IP solution an edge between free requests in the connected component of marked edges. the marked paths in the tree are then used to upper bound the distance cost (in the time-augmented metric space) of the match between the two endpoints in the tree 

Secondly, due to the online nature of the algorithm, one cannot assume two requests are matched upon the arrival of the second and hence, more delay cost can be incurred. This additional delay cost is not accounted for by the time-augmented metric space. By growing the moats at the rate of delay accumulation, Bienkowski et al~\cite{pd} manage to use the value of the dual to bound the additional delay cost incurred as well.

\paragraph{Introducing Concave Delay to the Framework}
As we consider the problem with a \emph{concave} delay cost function nevertheless, we need to introduce the concept of a \emph{concave time-augmented metric space}. 

\begin{definition}[Concave Time-Augmented Metric Space]
Given a metric space $D = (W,\dist)$ we define the \emph{concave time-augmented metric space} as $D_{CT} = (W \times \mathbb{R},\timedistc)$ where $\timedistc: 2^R \rightarrow \mathbb{R_+}$ is a distance function defined as 
\begin{align*}
	\timedistc(u,v) = \dist(\pos(u), \pos(v)) + f(|\atime(u) - \atime(v)|).
\end{align*}
Where $u,v \in W \times \mathbb{R}$ and $f$ defines a concave monotone non-decreasing function. In other words, a concave time-augmented metric space is a metric that is the Cartesian product of a metric space $D$ and a the time axis, such that the distance between two points is defined as the difference in position in addition to a concave function of the difference in arrival time of the points. 
\end{definition}
It remains to prove that the resulting space is a valid metric space. To this end we show that the new distance function satisfies triangle inequality. 
\begin{claim}
    \label{clm:triangle}
    For all $u,v,w \in R$, $\timedistc(u,v) \leq \timedistc(u,w) + \timedistc(w,v)$. In other words, the function $\timedistc$ satisfies triangle inequality.
\end{claim}
\begin{proof}
    Let $f$ be the concave function associated with $\timedistc$. From the subadditivity of a concave function, we deduce that $f$ satisfies that for all $u,v,w \in R$,
    \begin{align}
    \label{eq:subadditive}
        f(|\atime(u) - \atime(v)|) \leq f(|\atime(u) - \atime(w)|) + f(|\atime(w) - \atime(v)|).
    \end{align}
    Since $D$ is a valid metric space, it follows that $\dist$ satisfies triangle inequality. Therefore, we have
    \begin{align}
    \label{eq:triangle}
        \dist(\pos(u), \pos(v)) \leq \dist(\pos(u), \pos(w)) + \dist(\pos(w), \pos(v)).
    \end{align}
    From equations~\eqref{eq:subadditive} and~\eqref{eq:triangle} we conclude that
    \begin{align*}
        \timedistc(u,v) \leq \timedistc(u,w) + \timedistc(w,v).
    \end{align*}
\end{proof}
Solving the offline MPMD with concave delays problem in the metric space $D$ is equivalent to solving the offline min-cost perfect matching problem without delay in $D_{CT}$. This equivalence holds due to the fact that an optimal offline solution will always match a pair of requests upon the arrival of the second request. This, however, does not hold for the online version of the problem, where an additional delay cost may be incurred after the arrival of the second request.

If the reader is familiar with the Primal-Dual algorithm by Bienkowski et al.~\cite{pd}, they can skip to Section~\ref{ssec:ConcaveAlgorithm}.

\subsubsection{An Existing Primal-Dual Algorithm for MPMD with Linear Delays}Bienkowski et al.~\cite{pd} use the moat-growing framework above to design a linear competitive primal-dual algorithm for the problem with linear delay. For the benefit of the reader, we give a brief outline of their algorithm below. 

The algorithm maintains a family of \emph{active} sets such that, at any time, a request belongs to a single active set. They define a set to be \emph{growing} if it is both active and contains an odd number of requests. When a set is both active and growing, the corresponding dual variable grows at the rate of delay accumulation of the requests inside the set. Visually, this can be interpreted as growing moats around the requests in a time-augmented metric space. Because the delay functions are uniform linear, all requests accumulate delay at the same rate at all times, regardless of their arrival time. When a constraint becomes tight, the event can be visualised as the boundary of two moats around the sets $S_1$ and $S_2$ touching and, consequently, the accumulation of all moats around the sets containing the two endpoints will have covered the edge. When this happens, the algorithm marks the edge and performs a maximum cardinality matching between the unmatched requests in the active sets containing the endpoints. 

Both the proof of correctness and the analysis of~\cite{pd} heavily relies on the fact that all requests accumulate delay at the same rate at all times.

\subsection{An Algorithm for MPMD with Concave Delay}
\label{ssec:ConcaveAlgorithm}
In this subsection we briefly address the challenges involved with using the moat-growing framework to design competitive online algorithms for the problem with concave delays. We then in introduce the integer and linear programs used, and formally describe our algorithm. To prove Theorem~\ref{thm:concave} we give a proof of correctness and competitive analysis of the algorithm. 

\paragraph{Challenges in Moving from Linear to Concave Delay}In Bienkowski et al.'s approach to the problem with uniform linear delay~\cite{pd}, the rate of delay cost accumulation at any time is the same for all request (due to the properties of uniform linear delay). Consequently, the rate of growth of the dual variable associated with any set of requests is well-defined. However, when the uniform delay functions are concave, at any time, requests can accumulate delay cost at different rates (depending on their arrival times). The rate of growth of the dual variable of a set of requests is thus no longer well-defined. The main challenge is to grow the dual variables at a rate such that the dual solution remains feasible, while ensuring one can bound the competitive ratio of the algorithm in terms of the value of the dual solution produced. 

To achieve this, we grow the dual variable associated with a set $S$ each iteration by the maximum possible amount such that, if a request $v$ arrived at the same position in the original metric space (i.e. $\pos(v) = \pos(u)$) for any $u \in S$, the dual constraint associated with the new edge $(u,v)$ would be instantly tight, but not violated. This allows us to ensure no dual constraints are violated throughout the execution of the algorithm. Furthermore, at all times, for at least one request in the set, the sum of dual values, associated with sets that contain the request, is equal to the total delay this request would have incurred if left unmatched. In the analysis of our algorithm, we show that the additional online delay cost incurred by any request inside the set can therefore be upper bounded by the total growth of the dual variables.

\subsubsection{Primal-Dual Formulation}
We adapt the IP, LP and DP from the moat-growing framework to the MPMD with concave delay setting below. We abuse notation and use $R$ to denote both the set of requests in the original metric space and the corresponding points in the concave time-augmented metric space. We also use $u$ to denote both the request and the corresponding point in the concave time-augmented metric space. Similar to Bienkowski et al., let the function $f$ be the \emph{surplus} of a set $S$ of requests. Recall that the surplus of a set of requests is equal to 1 if the cardinality of the set S is odd, and 0 otherwise\footnote{Note that $\Sur(S) = 1$ indicates that at least 1 edge must leave the set in order to have a perfect matching of all requests.}. Furthermore, we define the cost function $c$ associated with an edge $(u,v)$ to be the distance between its two endpoints in the concave time-augmented metric space (i.e. $\timedistc(u,v))$. Similar to Bienkowski et al. we denote this cost as $\optcost((u,v))$.\\

Integer program (IP):
\begin{alignat*}{2}
  & \text{minimize: } & & \sum_{j=1} \optcost(e) \cdot x_{e} \\
   & \text{subject to: }& \quad & \sum_{\mathclap{{e \in \delta(S)}}}
                \begin{aligned}[t]
                    x_{e} & \geq \Sur(S), &\qquad \quad \forall S \subseteq R\\[3ex]
                  x_{e} & \in \{0,1\}, &\qquad \quad \forall e \in E
                \end{aligned}
\end{alignat*}

Linear program (LP):
\begin{alignat*}{2}
  & \text{minimize: } & & \sum_{j=1} \optcost(e) \cdot x_{e} \\
   & \text{subject to: }& \quad & \sum_{\mathclap{{e \in \delta(S)}}}
                \begin{aligned}[t]
                    x_{e} & \geq \Sur(S), &\qquad \quad \forall S \subseteq R\\[3ex]
                  x_{e} & \geq 0, &\qquad \quad \forall e \in E
                \end{aligned}
\end{alignat*}

Dual program (DP):
\begin{alignat*}{2}
  & \text{maximise: } & & \sum_{S \subseteq R} y_S \cdot \Sur(S) \\
   & \text{subject to: }& \quad & \sum_{\mathclap{S:e \in \delta(S)}}
                \begin{aligned}[t]
                    y_S & \leq \optcost(e), &\qquad \quad \forall e \in E\\[3ex]
                  y_S & \geq 0,	&\qquad \quad \forall S \subseteq R
                \end{aligned}
\end{alignat*}
\subsubsection{Description of Algorithm}
The high-level idea behind our algorithm is to apply the primal-dual algorithm from the moat-growing framework in a concave time-augmented metric space. Similar to Bienkowski et al.~\cite{pd}, at any time $t$, the algorithm partitions the requests that have arrived into \emph{active sets} by maintaining a mapping $A: R \rightarrow 2^R$ such that, at any time, a request $u \in R$ belongs to a single active set. If the surplus of the active set is strictly positive we say the active set is \emph{growing}, else, it is \emph{non-growing}\footnote{This notation is not the standard notation in the moat growing framework but rather builds on the notation introduced in Bienkowski et al.~\cite{pd}}. 
For any active growing set $S$, during each iteration, our algorithm increases its corresponding dual variable by the maximum amount such that, for any request $u \in S$, if a new request $v$ arrives at $\pos(v) = \pos(u)$ in the original metric space, no dual constraints would be violated. 

Let $S$ be an active growing set. During iteration $i$, we calculate the amount of growth of a set as follows. For each request $u$, the minimum distance between itself and any other point will be the distance to a new request $v$ that arrives at the same position in the original metric space. In this case, $\optcost(u,v) = f(|\atime(v) - \atime(u)|)$. Therefore, the sum of all dual variables associated with sets that contain $u$ must be less than this to maintain the feasibility of the constraint associated with the edge. We refer to this value as $\reqGrowth (u)$. The value is calculated as follows for any given time $t$:
\begin{align*}
    \reqGrowth(u) = f(|t - \atime(u)|) - \sum_{S^\prime: u \in S^\prime}y_{S^\prime}
\end{align*}

\begin{algorithm}[H]
	\SetAlgoLined
    \textbf{\emph{Request arrival event:}}\\
	\If{a new request u arrives}{
		$A(u) = \{u\}$
	
	}
    \textbf{\emph{Tight constraint event:}}\\
	\While{there exists a tight dual constraint for $(u,v)$ s.t. $A(u) \neq A(v)$}{
		mark $(u,v)$
		\\$S = A(u) \cup A(v)$\\
		\For{all $w\in S$}{
			$A(w) = S$
		}
		Perform a maximum cardinality matching between unmatched requests in $S$.
	}
    \textbf{\emph{None of the above events occur:}}\\
	
		\For{all active growing sets $S$}{
			increase $y_S$ by $\displaystyle \min_{x \in S}\{\reqGrowth(x)\}$\\
	   }
	\caption{Primal-Dual Algorithm for Concave MPMD}
\end{algorithm}
\leavevmode\newline 
We denote the time at which the algorithm terminates as $T$. Let $y_S^T$ be the value of the dual variable corresponding to the set $S$ upon termination of the algorithm. 

Observe that, during each iteration $i$, (corresponding to the $i$th timestep), for any $S \subseteq R$ and $u \in S$, the sum of all dual variables corresponding to sets that contain $u$, is at most the total delay incurred by the request if it had remained unmatched. 
\begin{observation}
    \label{obs:feasible}
    For any timestep $i$ and $S \subseteq R$, for any $u\in S$, $\sum_{S^\prime: u \in S^\prime}y_{S^\prime} \leq f(i - \atime(u))$
\end{observation}
Furthermore, for at least one request $x$, the sum of the values of the dual variables of all sets that contain it is equal to the total delay cost it would have incurred so far if left unmatched.
\begin{observation}
    \label{obs:tight}
    For any timestep $i$ and $S \subseteq R$, there exists $x \in S$ such that $f(i - \atime(x)) = \sum_{S: x \in S} y_S$.
\end{observation}
Finally, observe that a dual variable will only be non-zero if the set has an odd number of requests.
\begin{observation}
	\label{obs:sur}
	Upon termination of the algorithm, for any $S \subseteq R$, $y_S^T > 0$ if $\Sur(S) \geq 1$.
\end{observation}

\subsubsection{Correctness}
Let $y^t_S$ denote the value of the dual variable $y_S$ at time $t$. In order to relate the cost of the matching produced by the algorithm to the value of the dual solution produced, the dual needs to be feasible. Furthermore, to prove the correctness of the algorithm, we need to show that, upon termination, it produces a perfect matching of the requests. We start by proving the dual solution maintained by the algorithm is feasible. 

\begin{lemma}
	\label{lem:feasible}
	At any time $i$, the values $y^i_S$ constitute a feasible solution to the dual.
\end{lemma}

\begin{proof} We show that no dual constraint is violated throughout the execution of the algorithm. When the algorithm starts, all dual variables are initialised to 0, which satisfies all constraints. A dual variable is only increased if the corresponding set of requests is both active and growing. By definition, it can only be active and growing if, for all edges that cross the cut of $S$, the corresponding dual constraint is not yet tight. When a constraint becomes tight, the set is deactivated and the edge that corresponds to the tight constraint becomes an internal edge in the new active set. Since its two endpoints from now on belong to the same active set, its corresponding dual constraint can no longer increase. 

The only remaining case to consider is when a new request arrives. Upon arrival of a new request $u$, all new dual variables associated with sets that contain $u$ are initialised to 0. For each existing constraint, corresponding to edges not incident to $u$, new dual variables may be added to the set. However, since these variables are initialised with the value 0, they do not contribute towards the sum, preserving the feasibility of the constraint. For all new edges (edges whose endpoint is $u$), we show the new constraint is feasible, i.e., $\sum_{S:(u,v)\in \delta(s)}y^t_S \leq \optcost\big((u,v)\big)$. 
 \begin{align*}
 	\sum_{S:(u,v) \in \delta(S)}	 y^t_S &= \sum_{S:v\in S \wedge u \notin S} y^t_S + \sum_{S:v\notin S \wedge u \in S} y^t_S\\
 	&= \sum_{S:v\in S \wedge u \notin S} y^t_S\\
 	&\leq \sum_{S:v\in S} y^t_S\\
 	&\leq f(t - \atime(v)) \qquad &\text{(by Observation~\ref{obs:feasible})} \\
 	&= f(\atime(u) - \atime(v))\\
 	&\leq \optcost\big((u,v)\big). \qedhere
 \end{align*}
\end{proof}

Finally, to establish that, upon terminating, the algorithm returns a perfect matching of all requests, we use the following lemma, which follows standard arguments in moat-growing algorithms.

\begin{lemma}
    \label{lem:perfmatch}
	For MPMD with concave delays, the algorithm returns a perfect matching of all requests.
\end{lemma}

\begin{proof}
	Let $f$ be a continuous, monotone non-decreasing, concave function. For the sake of contradiction, suppose, upon termination of the algorithm, $u$ remains unmatched. $u$ must at all times belong to some active set $A(u)$. Since $u$ is unmatched $A(u)$ remains both active and growing. It follows that the dual variable $y_{A(u)}$ will continue to increase indefinitely. All dual variables have a positive coefficient in the dual objective function. By Lemma~\ref{lem:feasible}, the dual solution maintained by the algorithm is feasible. It follows the dual must be unbounded and no bounded feasible primal solution exists. Since all distances in the metric space are finite, there exists a finite solution to the primal program. This constitutes a contradiction.
\end{proof}

\subsubsection{Cost Analysis}
Having proven that the algorithm produces a valid matching $M$ and maintains a feasible dual solution, we now prove Theorem~\ref{thm:concave} by bounding the cost of the matching produced by ALG in terms of the value of the dual solution it produces along the way.
 
We start by defining the cost of $M$, which consists of the distance between matched requests in the original metric space, plus the delay cost incurred by each request. We denote the cost of $M$ as $\cost(M)$, and the cost of a single edge $(u,v) \in M$ as $\cost((u,v))$. Let $t_{(u,v)}$ be the time at which $u$ and $v$ are matched. Without loss of generality, we assume $\atime(v) \geq \atime(u)$. For all $(u,v)$ in $M$ we define the cost of a match as follows.
\begin{align*} 
	\cost((u,v)) = \dist(u,v) + f(t_{(u,v)}-\atime(u)) + f(t_{(u,v)}-\atime(v)).
\end{align*}
In order to bound the cost of a match in terms of the value of the dual solution, we first express the cost in terms of the distance between the two requests in the concave time-augmented metric space. We do this by first adding and subtracting $f(\atime(v)-\atime(u))$, and then rearranging the terms.
\begin{align*} 
	\cost((u,v)) &= \big(\dist(u,v) + f(\atime(v)-\atime(u))\big) \\
                    & + \big(f(t_{(u,v)}-\atime(u)) - f(\atime(v)-\atime(u))\big)\\ 
                    & + f(t_{(u,v)}-\atime(v))\\
    &= {\timedistc(u,v)} - f(\atime(v)-\atime(u)) + 2\big(f(t_{(u,v)}-\atime(u)).
\end{align*}
Using the subadditive property of a concave function, we note that $f(|t_{(u,v)}-\atime(u)|) - f(|\atime(v)-\atime(u)|) \leq f(|t_{(u,v)}-\atime(v)|)$. Hence, we upper bound the cost as follows:
\begin{align*} 
	\cost((u,v)) &\leq {\timedistc(u,v)} + 2 \cdot f(t_{(u,v)}-\atime(v))\\
	&\leq {\optcost\big((u,v)\big)} + 2 \cdot f(t_{(u,v)}-\atime(v)).
\end{align*}
To analyse the cost of the matching we will, for all $(u,v) \in M$, bound both $\optcost\big((u,v)\big)$ and $f(t_{(u,v)}-\atime(v))$ in terms of the value of the dual solution. We start by bounding $\optcost\big((u,v)\big)$ in terms of the cost of the marked edges, corresponding to tight dual constraints. We prove that, for any edge $(u,v) \in M$, there exists a path $P$ of marked edges between $u$ and $v$.

We start by proving the following lemmas, which address the state of the marked edges throughout the execution of the algorithm (note that these lemmas follow standard arguments for moat-growing algorithms).

\begin{lemma}
    \label{lem:cut}
    At any point in time, there exist no marked edges that cross the cut of an active set.
\end{lemma}
\begin{proof}
    When the algorithm commences, all active sets are singletons and are trivially spanned by the an empty set of marked edges. An edge $(u,v)$ is only ever marked in the event of a tight constraint. In the event of a tight constraint, the two active sets $A(u)$ and $A(v)$ containing $u$ and $v$ merge, deactivating $A(u)$ and $A(v)$, and activating the set of their union. Since $u$ and $v$ from this point on belong to the same active set, the marked edge will never cross the cut of an active set. 
\end{proof}

An edge $(u,v)$ can only be added to $M$ in the event of a tight constraint where two active sets merge. Hence, to prove there exists a path of marked edges between $u$ and $v$ we need to argue that there exists a unique path of marked edges between any two requests inside an active set. 

\begin{lemma}
    \label{lem:spanningtree}
    During any iteration $i$ of the algorithm, the subset of marked edges that are entirely contained in an active set $S$ form a spanning tree of all the requests in $S$.
\end{lemma}
\begin{proof}
    We show the property holds by induction on the number of iterations of the algorithm. During the first iteration, the empty set forms a valid spanning tree of the active sets that are singletons. We assume the property holds for all sets up to some iteration $k$ during the execution of our algorithm. If, at time $k+1$, a tight constraint event occurs, corresponding to some edge $(u,v)$, the active sets $A(u)$ and $A(v)$ merge. By our inductive hypothesis, the marked edges entirely contained inside $A(u)$ and $A(v)$ form spanning trees of the requests inside each respective set. From Lemma~\ref{lem:cut} we know that there are no marked edges that cross between the two sets. Therefore, the union of the spanning trees in $A(u)$ and $A(v)$, along with the marked edge that corresponds to the tight dual constraint, form a valid spanning tree of the new active set.
\end{proof}

We are now ready to argue that, for any $(u,v)\in M$, there exists a path $P$ of marked edges between them.

\begin{lemma}
    \label{lem:path}
    For any $(u,v) \in M$, upon termination of the algorithm, there exists a path $P$ of marked edges between them.
\end{lemma}
\begin{proof}
    By construction of the algorithm, if $(u,v)\in M$, upon termination of the algorithm, $A(u) = A(v)$. By Lemma~\ref{lem:spanningtree} there must exist a path of marked edges between $u$ and $v$ (which is entirely contained inside $A(u)$).
\end{proof}

Next, we bound the cost of $P$ in terms of the value of the dual solution. We first note that the dual constraint corresponding to any marked edge must be tight. Hence, for all marked edges $e$, $\optcost(e) = \sum_{S: e \in \delta(S)}$. In order to express this in terms of the value of the dual solution, we need to be able to upper bound the number of times the path $P$ crosses the cut of any set $S$. For this purpose, we use the following lemma. 

\begin{lemma}
	\label{lem:geq2}
    For any path of marked edges $P$ and active set $S \subseteq R$, $|\delta(S) \cap P| \leq 2$.
\end{lemma}
\begin{proof}
For the sake of contradiction, assume there exists a path $P$ that crosses the boundary of a set $S$ more than twice. Let us denote the endpoints of the path as $u$ and $v$, We direct the path from $u$ to $v$\footnote{Note that $u$ may be in $S$.}. Since the path crosses the boundary of $S$ more than twice, there must exists at least one edge $(w_1, w_2) \in P$ such that $w_1 \in S$ and $w_2 \notin S$, and at least one edge $(w_3, w_4) \in P$ such that $w_3 \notin S$ and $w_4 \in S$. By Lemma~\ref{lem:spanningtree}, there must exist a path of marked edges between $w_1$ and $w_4$ that is entirely contained inside $S$. This indicates the existence of a cycle of marked edges. However, from Lemmas~\ref{lem:cut} and ~\ref{lem:spanningtree} we can deduce that the set of marked edges constitutes a forest. This is a contradiction.
\end{proof}

We are now ready to bound $\optcost\big((u,v)\big)$ for any $(u,v) \in M$ in terms of the value of the dual solution. We do this using the following lemma. We denote the value of the dual variables upon termination of the algorithm as $y_S^T$.

\begin{lemma}
    \label{lem:optcost}
    Upon termination of the algorithm, for any $u,v \in R$, connected by a path of marked edges $P$, $\optcost\big((u,v)\big) \leq 2 \cdot \sum_{S \subseteq R} y_S^T \cdot \Sur(S)$.
\end{lemma}
\begin{proof}
By triangle inequality, we can upper bound $\optcost\big((u,v)\big)$ by cost of the edges on the path $P$ between $u$ and $v$. 
	\begin{align*} 	
		\optcost\big((u,v)\big) &\leq \sum_{e \in P} \optcost(e)\\
		&=\sum_{e \in P} \sum_{S:e \in \delta(S)} y_S^T \\ 
		&=\sum_S |\delta(S) \cap P| \cdot y_S^T \\
		&\leq \sum_S |\delta(S) \cap P| \cdot y_S^T \cdot \Sur(S) \qquad \text{\emph{(By Observation~\ref{obs:sur}.)}}\\
		&\leq 2 \cdot \sum_S \cdot y_S^T \cdot \Sur(S) \qquad \text{\emph{(By Lemma~\ref{lem:geq2}.)}}.
	\end{align*}
 The first equality follows from the fact that marked edges correspond to tight dual constraints.
\end{proof}

We have bounded $\optcost\big((u,v)\big)$ in terms of the value of the dual solution. It remains to bound $f(t_{(u,v)}-\atime(v))$. To this end, we introduce the following Lemma. We assume without loss of generality that $\atime(v) \geq \atime(u)$.
\begin{lemma}
	\label{lem:add_delay}
	Let $(u,v)\in M$ and let $t_{(u,v)}$ be the time at which $u$ and $v$ are matched. $f(t_{(u,v)} -\atime(v))\leq \sum_{S \subseteq R} y_S^T \cdot \Sur(S)$.
\end{lemma}
\begin{proof}

From Observation~\ref{obs:tight} we know that $A(v)$ contains at least one request $x$, for which it holds that, at time $t_{(u,v)}$, $\sum_{S:x \in S} y_S = f(t_{(u,v)}-\atime(x))$. Since $f$ is a concave function,
\begin{align*} 
	f(|t_{(u,v)}-\atime(v)|) \leq f(|\atime(v) - \atime(x)|) + f(|t_{(u,v)}-\atime(x))|).
\end{align*}
Since $v, x \in A(v)$, we know there exists a path of marked edges, which correspond to tight dual constraints, between $v$ and $x$. Therefore, by Lemma~\ref{lem:optcost}, we can upper bound $\optcost((v,x))$ by $2 \cdot \sum_{S \subseteq V} y_S^{T} \cdot \Sur(S)$. We use this to upper bound $f(|\atime(v) - \atime(x)|)$ as follows:
\begin{align*} 
	f(|\atime(v) - \atime(x)|) &\leq f(|\atime(v) - \atime(x)|) + \dist(\pos(v),\pos(x)) \\
	&= \optcost((v,x))\\
	&\leq 2 \cdot \sum_{S \subseteq V} y_S^{T} \cdot \Sur(S)
\end{align*}
Recall that $x$ is the request for which it holds that, at any time $t$, $\sum_{S:x \in S} y_S^{t} = f(t-\atime(x))$. By observation~\ref{obs:sur} we know that $\sum_{S \subseteq V} y_S^{t} \cdot \Sur(S) \geq f(t - \atime(x)) $. It follows that
\begin{align*} 
		f(t_{(u,v)}-\atime(v)) &\leq 2 \cdot \sum_{S \subseteq V} y_S^{t_{(u,v)}} \cdot \Sur(S) + \cdot \sum_{S \subseteq V} y_S^{t_{(u,v)}} \cdot \Sur(S) \\
		&= 3 \cdot \sum_{S \subseteq V} y_S^{t_{(u,v)}} \cdot \Sur(S)
\end{align*}
\end{proof} 
Recall that $D_{value}$ is the value of the dual solution, and let $OPT$ be the cost of the optimal IP solution. From Lemmas~\ref{lem:optcost} and~\ref{lem:add_delay} we can now deduce 
\begin{align*} 
	\cost(M) &\leq 8 \cdot \sum_{(u,v) \in M}\sum_{S \subseteq V} y_S^{T} \cdot \Sur(S)\\
    &= 8 \cdot \sum_{(u,v) \in M}D_{value}\\
    &\leq 8 \cdot \sum_{(u,v) \in M}OPT\\
    &= 8m \cdot OPT
\end{align*}
, which completes the proof of Theorem~\ref{thm:concave}.
 
\bibliographystyle{plainurl}
\bibliography{main.bib}
\end{document}